%% file: main.tex
\documentclass[review=false,nonacm]{acmart}
\usepackage[T1]{fontenc}
\usepackage{booktabs} 
\usepackage[ruled]{algorithm2e} 

\SetAlFnt{\small}
\SetAlCapFnt{\small}
\SetAlCapNameFnt{\small}
\SetAlCapHSkip{0pt}
\IncMargin{-\parindent}
\newcommand{\extraqed}{}
\usepackage{amsmath}
\usepackage{breqn}
\usepackage{hyperref}
\usepackage{wrapfig}
\usepackage{dsfont, enumitem}

\newtheorem{assumption}{Assumption}
\newtheorem{observation}{Observation}
    \newtheorem{theorem}{Theorem}[section]
    \newtheorem{corollary}[theorem]{Corollary}
    \newtheorem{lemma}[theorem]{Lemma}
    \newtheorem*{lemma*}{Lemma}
    \newtheorem{example}[theorem]{Example}
    \newtheorem*{example*}{Example}
    \newtheorem{definition}[theorem]{Definition}
    
    \newtheorem*{theorem*}{Theorem}
    \newtheorem*{restatement*}{Restatement}

\newcommand{\ashish}[1]{\todo[color=pink]{AG: #1}}

\newcommand{\geoff}[1]{\todo[color=brown]{GR: #1}}
\newcommand\LVR[1]{\begingroup \color{black} #1 \endgroup}

\title{Finding the Right Curve:\\Optimal Design of Constant Function Market Makers}
\input{authors}

\begin{abstract}
\input{abstract}
\end{abstract}

\begin{document}
\begin{titlepage}
\maketitle
\end{titlepage}

\input{intro}

\input{related_work}
\input{preliminaries}

\input{model}

\input{optimization}

\input{beliefs}

\input{overall_profit_and_loss}

\input{conclusion}

\bibliographystyle{splncs04}
\bibliography{main}

\newpage

\appendix

\input{cost_fn_apx}

\section{Continuous trade size distribution} \label{sec:cont}
\input{continuous_trading_main_paper}

\section{Omitted Proofs} \label{sec:proofs}
\input{prelim_obs_proofs}

\input{model_proofs}
\input{cor_equiv_class}
\input{varphi_proof}
\input{finite_obj_proof}
\input{basic_opt_proof}
\input{cor_yhat_defined}
\input{cor_lp_equiv}
\input{cor_inverse_proof}
\input{addbeliefs_proof}
\input{univ3_proof}
\input{balancer_proof_apx}

\input{lmsr_proof}
\input{loss_proof}
\input{bid_ask_proof}

\end{document}

%% file: authors.tex
\newcommand{\geoffemail}{geoff.ramseyer@cs.stanford.edu}
\newcommand{\mohakemail}{mohakg@stanford.edu}
\newcommand{\ashishemail}{ashishg@stanford.edu}
\newcommand{\dmemail}{}

\author{Mohak Goyal*}
\affiliation{%
  \institution{Stanford University}
  \city{Stanford}
  \state{California}
  \postcode{94305}
  \country{USA}}
\email{\mohakemail}

\author{Geoffrey Ramseyer*}
\affiliation{%
  \institution{Stanford University}
  \city{Stanford}
  \state{California}
  \postcode{94305}
  \country{USA}}
\email{\geoffemail}

\author{Ashish Goel}
\affiliation{%
  \institution{Stanford University}
  \city{Stanford}
  \state{California}
  \postcode{94305}
  \country{USA}}
\email{\ashishemail}

\author{David Mazières}
\affiliation{%
  \institution{Stanford University}
  \city{Stanford}
  \state{California}
  \postcode{94305}
  \country{USA}}
\email{\dmemail}

\thanks{* denotes alphabetical order.  This work is supported by the Future of Digital Currency Initiative (FDCI), Stanford University, and the Office of Naval Research, award number N000141912268.}

%% file: abstract.tex
Constant Function Market Makers (CFMMs) are a tool for creating exchange markets, have been deployed effectively in prediction markets, 
and are now especially prominent
in the Decentralized Finance ecosystem.
We show that for any set of beliefs about future asset prices, an optimal CFMM trading function exists that maximizes the fraction of trades that a CFMM can settle.
We formulate a convex program to compute this optimal trading function.
This program, therefore, gives a tractable framework for market-makers to compile their belief function on the future prices of the underlying assets into the trading function of a maximally capital-efficient CFMM.

Our convex optimization framework further extends
to capture the tradeoffs between fee revenue, arbitrage loss, and opportunity costs of
liquidity providers. Analyzing the program shows how the consideration of profit and loss leads to a qualitatively different optimal trading function.
Our model additionally explains the diversity of CFMM designs that appear in practice.
We show that careful analysis of our convex program enables inference of a market-maker's
beliefs about future asset prices, and show that these beliefs mirror the folklore intuition
for several widely used CFMMs.  Developing the program requires a new notion of 
the liquidity of a CFMM, and the core technical challenge is in the analysis
of the KKT conditions of an optimization over an infinite-dimensional Banach space.

%% file: intro.tex
\section{Introduction}

Agents in any economic system need to be able to exchange one asset for another efficiently.
%
Some assets are frequently traded by many market participants, and for these assets, 
a seller offering a reasonable price can likely find a buyer quickly and vice versa.  However, not every pair of assets
is traded frequently, and sellers in these markets might have to wait a long time to find a buyer or accept a
highly unfavourable price.
The role of a \emph{market-maker} is to fill this gap --- to facilitate easy and rapid trading between pairs of assets
for which otherwise there is very little trading activity.  Market-makers trade in both directions on the market, buying and selling
assets when traders arrive at the market \cite{amihud1980dealership}.  In this sense, market-makers facilitate asynchronous trading
between buyers and sellers, thereby increasing the market \emph{liquidity} between two assets.

Our topic of study is a subclass of automated market-making strategies known as {\it Constant Function Market Makers} (CFMMs). A CFMM maintains reserves of two assets $X$ and $Y,$ provided by a so-called \emph{liquidity provider (LP),} and makes trades according to a predefined \emph{trading function} $f(x,y)$ of its asset reserves (the eponymous "constant function"); specifically, a CFMM accepts a trade $(\Delta x, \Delta y)$ from reserves $(x,y)$ to $(x - \Delta x, y + \Delta y)$ 
if and only if $f(x - \Delta x, y + \Delta y)=f(x, y)$.
CFMMs earn revenue by charging a small commission on each trade (i.e. creating a bid-ask spread) but are subject to several associated
expenses \cite{amihud1986asset}, such as the costs of maintaining the asset inventory and adverse selection by arbitrageurs (i.e., stale quote sniping). The loss of the LP relative to the counterfactual strategy of ``buy-and-hold'' is referred to as the ``divergence loss'' \cite{milionis2022automated}.

Automated market-making has long been an important topic of study \cite{aoyagi2020liquidity,gerig2010automated,othman2013practical},
but CFMMs have recently become some of the most widely used exchanges
\cite{uniswapv2,uniswapv3,balancer,egorov2019stableswap} within the modern 
Decentralized Finance (DeFi) ecosystem \cite{werner2021sok}. The success of CFMMs in DeFi is primarily due to their ability to run via smart contracts \cite{mohanta2018overview} with a fairly low computation requirement on blockchains. CFMMs also reduce the barrier to entering the liquidity provision business or ``market-making'' \cite{ammdemocratize}.
CFMMs have also been widely deployed in prediction markets
as a method for aggregating opinions \cite{hanson2007logarithmic,chen2010new}.
For completeness, we describe the precise translation from market scoring rules studied in the prediction markets literature to CFMMs in Appendix \S \ref{apx:cost_fn}.


\begin{example}[Real-world CFMMs]
~
	\begin{enumerate}
		\item
			The decentralized exchange Uniswap \cite{uniswapv2} uses the product function $f(x,y) = xy$.

		\item 
			The Logarithmic Market Scoring Rule (LMSR) \cite{hanson2007logarithmic}, used extensively to design prediction markets, corresponds to a CFMM with trading function
			$f(x,y) = (1-e^{-x}) + (1 - e^{-y})$ \cite{univ3paradigm}.

		\item
			The trading function $f(x,y)=xe^y$ has powered automated storefronts in online games \cite{hyperconomy}.
%
	\end{enumerate}
\end{example}
%


Despite facilitating billions of US dollars worth of trade volume per day, a complete formal understanding of CFMM design trade-offs is missing in the literature. 
Our goal, therefore, is to explain what guides a CFMM designer to choose one
trading function over another. We provide an optimization framework which compiles a market-maker's beliefs on future prices into an optimal CFMM trading function, making substantial progress towards an important open problem \cite{timtweet}.


\input{our_results}

%% file: our_results.tex
\subsection{Our Contributions}

We develop a convex optimization framework that translates an LP's beliefs about future asset valuations into
an optimal choice of CFMM trading function.  We show that a unique trading function always maximizes an LP's
expectation of the fraction of trades that a CFMM can settle (\S\ref{sec:inefficiency}). 
Furthermore, our framework is versatile such that it can model
a wide variety of real-world concerns, \LVR{including fee revenue (\S\ref{sec:net_profit}), divergence loss (\S\ref{sec:divloss}), so-called ``Loss-Versus-Rebalancing'' \cite{milionis2022automated} (\S\ref{sec:lvr}), and models of price dynamics (\S\ref{sec:price_dyn}).} To the best of our knowledge, this is the first unified framework for analysing and optimizing CFMMs for various objectives for a given belief function and therefore provides a guide to prospective LPs interested in using a CFMM.

We model an LP's beliefs as a joint distribution on the future prices of the two assets with regard to a numeraire (\S \ref{sec:beliefs}). This belief could, for example, be generated from a price dynamics model.
As one might expect, the optimization problem concerned only with maximizing the fraction of trades settled depends only
on the distribution of the \emph{ratio} of prices.  However, expressing beliefs on future prices as a joint distribution
enables optimizing for profit and loss through the same framework.

We measure the liquidity of the CFMM trading function as the amount of capital implicitly allocated for market-making
at a given spot exchange rate.  
Specifically, the liquidity of a CFMM is the ratio between the size of a trade and the percentage change in
the spot exchange rate (\S \ref{sec:liquidity}). 

We analyze the steady-state dynamics of trade requests on a CFMM and arrive at a notion of ``CFMM inefficiency'' that approximates the probability that a CFMM cannot satisfy a trade request. CFMM inefficiency is a function of the inverse of the CFMM's liquidity.  Ultimately, we find that more complex objective functions considering an LP's profit and loss become linear combinations of CFMM liquidity and CFMM inefficiency.

Careful analysis of the KKT conditions of our convex program allows us to invert the problem; given an arbitrary CFMM trading function,
we can construct an explicit equivalence class of beliefs for which the given function is optimal.  
The main technique involves analysis of the KKT conditions of an optimization problem over an infinite-dimensional Banach space. We obtain closed-form solutions to the optimal CFMM designs for several important belief functions and objective functions. When not closed-form, the solution is still computationally tractable.

Our framework helps explain the choice of CFMM trading functions deployed in practice.  In many cases, the optimal CFMM revealed by our framework matches  the informal intuitions of practitioners. For example, the Uniswap V2 \cite{uniswapv2} protocol in DeFi was designed using the constant product $f(xy) = xy$ CFMM with the motivation 
that the available liquidity must be spread evenly across all exchange rates. In our framework, $f(xy) = xy$ is the optimal CFMM trading function for the uniform belief function with the objective of minimising the expected CFMM inefficiency.

\begin{figure}
\begin{center}
\begin{tabular}{|| c | c | c |   c||} 
 \hline
CFMM &Trading Function $f(x,y)$ & Liquidity $L(p)$ & Belief $\psi(p_X, p_Y)$\\ 
 \hline
Constant product & $xy$ & $\frac{1}{2} \sqrt{p}$ &  1 \\ 
 \hline
Constant weighted product & $x^{\alpha}y$ & $\frac{\alpha}{\alpha+1} p^{\frac{\alpha}{\alpha+1}}$ & $\left(\frac{p_X}{p_Y}\right)^{\frac{\alpha - 1}{\alpha + 1}}$ \\
 \hline
LMSR based & $1-e^{-x} + 1 - e^{-y}$ & $\frac{p}{1+p}$ & $\frac{p_X p_Y}{(p_X+p_Y)^2}$\\
 \hline
Lognormal prior based & As in Figure~\ref{fig:all_f} & $\sqrt{e^{\frac{-(\ln p)^2}{2\sigma^2}}}$ & $ \frac{p_X}{p_Y} \sim lognormal(0, \sigma^2)$\\
\hline
Black-Scholes based & As in Figure~\ref{fig:all_f} & Not closed form & As in equation~\ref{eq:BS}.\\
 \hline
\end{tabular}
\end{center}

\vspace{0.2cm}
  \begin{minipage}[c]{0.5\textwidth}
\includegraphics[width=\textwidth]{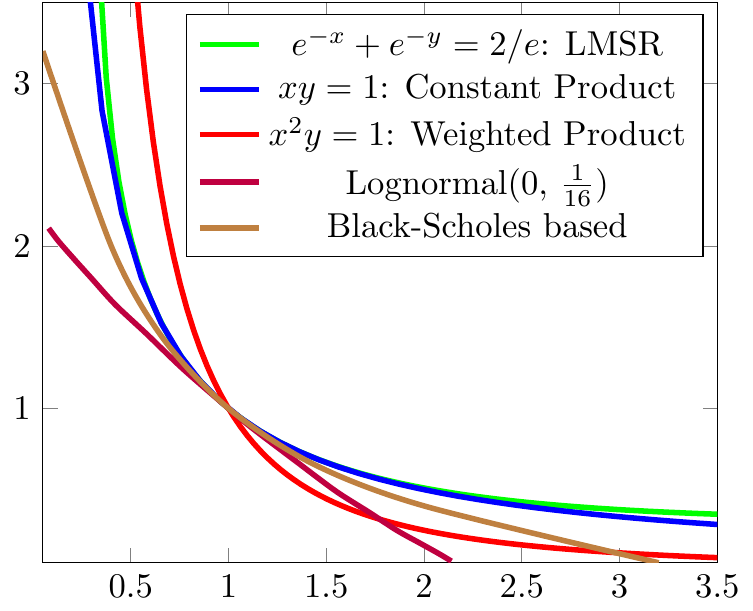}
  \end{minipage}\hfill
  \begin{minipage}[c]{0.5\textwidth}
    \caption{
       Some natural or widely used CFMM trading functions. 
       The lognormal belief function arises when we consider a snapshot of the Black-Scholes process at a future time. The entire Black-Scholes process can be considered (in expectation) for the purpose of our optimization framework via time-discounting (in the plot, the discounting parameter is 1). It can then be compiled into a single belief function as described in \S\ref{sec:price_dyn}. Our Python script in the Github repository \url{https://github.com/gramseyer/cfmm-liquidity-optimization} computes the belief function for the Black Scholes model. For any user-submitted belief function, our script also finds the optimal CFMM trading functions for minimizing CFMM inefficiency.
    }     \label{fig:all_f}
  \end{minipage}
\vspace{0.2cm}

%
%
\begin{tabular}{cc}
\begin{minipage}[t]{0.45\linewidth}\centering\includegraphics[width=0.70\linewidth]{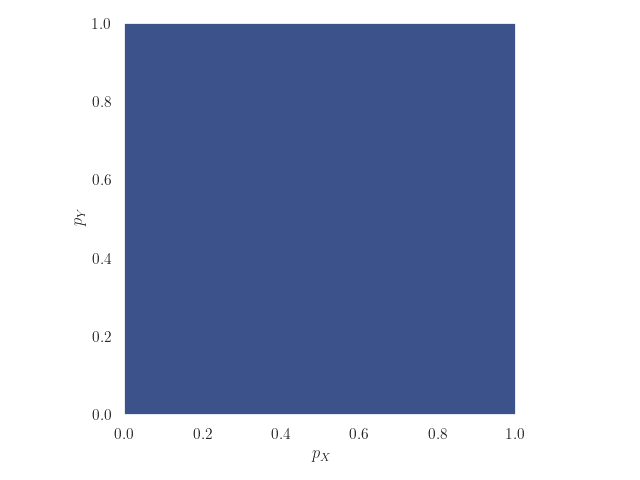} \caption{Constant belief function -- leads to the constant product market maker.} 
\end{minipage} & 
\begin{minipage}[t]{0.45\linewidth}\centering\includegraphics[width=0.70\linewidth]{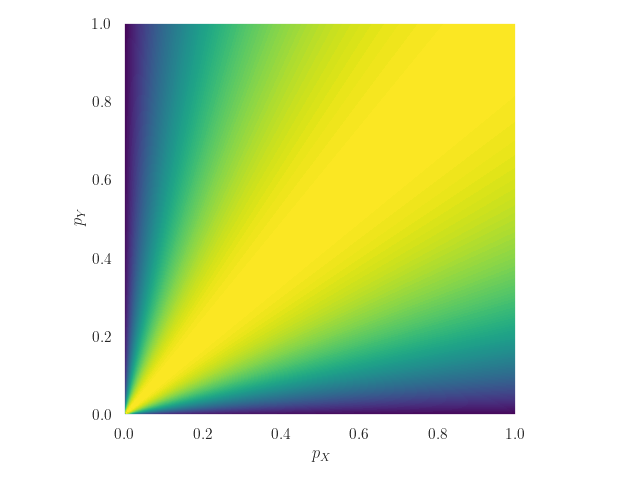} \caption{Belief function $\psi(p_X,p_Y) = \frac{p_X p_Y}{(p_X+ p_Y)^2}$ -- leads to the LMSR based market maker} 
\end{minipage}
\end{tabular}

\begin{tabular}{cc}
\begin{minipage}[t]{0.45\linewidth}\centering\includegraphics[width=0.70\linewidth]{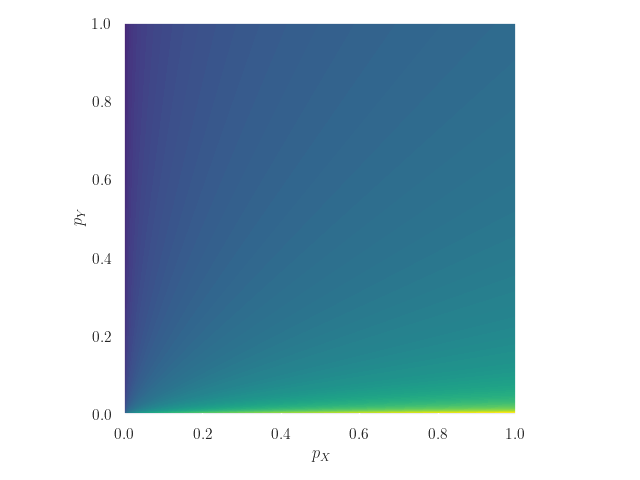} 
\caption{Belief function $\psi(p_X,p_Y) = \left(\frac{p_X}{p_Y}\right)^{1/5}$ -- leads to the weighted product market maker.} 
\end{minipage} &
\begin{minipage}[t]{0.45\linewidth}\centering\includegraphics[width=0.70\linewidth]{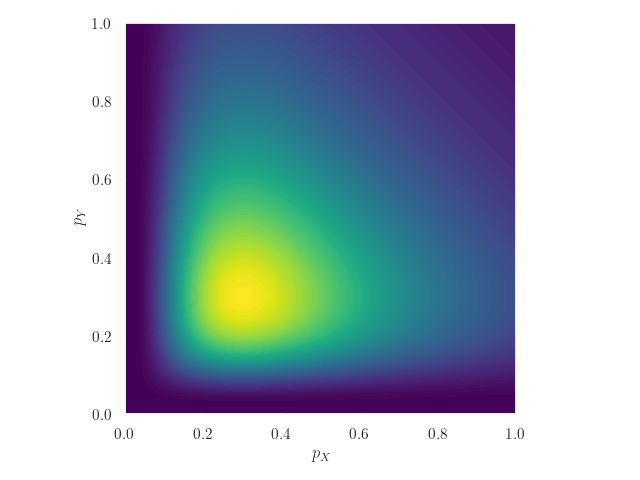} 
\caption{ Lognormal belief function -- leads to the Black-Scholes based CFMM after time-discounted aggregation of the belief function.}
\end{minipage} \\
\end{tabular}
\caption{Belief functions on future prices of the underlying assets relative to a numeraire. The plots are on beliefs defined on the range $(0,1]$ -- this is without loss of generality per Corollary~\ref{cor:equivclass}.}
\label{fig:psi} 
\end{figure}

Figure~\ref{fig:all_f} shows the trading functions and 
Figures~\ref{fig:psi} represent the belief functions for which the constant product, LMSR, constant weighted-product, and Black-Scholes-based CFMMs, respectively, are optimal for minimizing the CFMM inefficiency.\footnote{We provide, at the GitHub repositiry \url{https://github.com/gramseyer/cfmm-liquidity-optimization}, a Python script to generate the optimal CFMM trading functions for any user-defined belief function.}

\S \ref{sec:preliminaries} formally defines a CFMM and gives some basic properties. \S \ref{sec:model} studies the steady-state dynamics of a CFMM
and defines CFMM liquidity and inefficiency.  
\S \ref{sec:optimization} gives our convex optimization framework and analyzes its KKT conditions.  
\S \ref{sec:beliefs} studies the beliefs implicit behind real-world CFMM deployments.
Finally, \S \ref{sec:net_profit} shows how to add consideration for profit and loss to our framework,
and qualitatively studies how these considerations change the optimal trading function.
%
%

%% file: related_work.tex
\subsection{Related Work}
%
%
%
%
The closest line of work \cite{fan2022differential,neuder2021strategic,cartea2022decentralised,heimbach2022risks,bar2023uniswap} to our paper is the one which considers profit-maximizing market-making strategies which can be implemented via the Uniswap v3 \cite{uniswapv3} protocol. Additionally, \cite{neuder2021strategic,cartea2022decentralised,bar2023uniswap} design ``rebalancing'' strategies for the LPs, wherein they effectively modify the CFMM trading function periodically. In contrast, we consider designing CFMM trading functions from the first principles and do not use the Uniswap v3 framework. We also do not consider rebalancing the CFMM trading function in this work. A non-exhaustive list of papers in this line is:
\begin{itemize}[leftmargin = 0.3cm]
    \item Fan et al. \cite{fan2022differential}, study the question of maximizing risk-adjusted
profit for LPs while accounting for the gas fee for traders.  
Their model assumes that all trading on a CFMM occurs only in response to price movements on an external market (i.e. arbitrageurs realigning the CFMM spot price to the external market).
Their model suggests that risk-neutral
LPs must allocate all of their capital at a single price point (\S 4.2, \cite{fan2022differential}),  while ours better explains the choices of practitioners. 

\item Neuder et al. \cite{neuder2021strategic} 
study dynamic liquidity allocation strategies for risk-adjusted fee revenue maximization,
but do not consider the ``divergence loss'' incurred in the process.

\item Cartea et al. \cite{cartea2022decentralised} decompose the CFMM divergence loss into two components -- the convexity cost (loss due to arbitrage) and the opportunity cost (the cost of locking up capital). 
They give a stochastic optimal control-based closed-form strategy for a profit-maximizing LP.

\item 

Heimbach et al. \cite{heimbach2022risks} model liquidity positions on Uniswap V3 and perform a data-based analysis of the risks and returns of LPs as a function of the volatility of the underlying assets.

\end{itemize}


Similar to \cite{cartea2022decentralised}, Milionis et al. \cite{milionis2022automated} show that a part of the divergence loss corresponds to the market risk and can be hedged by a rebalancing strategy; the remainder of the divergence loss corresponds to the profit made by arbitragers trading against the CFMM -- they call this loss the LVR (loss-vs-rebalancing). When the variance of the price of $X$ relative to $Y$ is $\sigma^2,$ they show that the rate of accrual of LVR (what they call the instantaneous-LVR) is $\sigma^2 p^2 |x^{\prime}(p)|,$ where $x^{\prime}(p)$ denotes the rate of change of $x$ in the CFMM with respect to the price $p$ under perfect arbitrage. Since the LVR is a linear function of our notion of liquidity, our convex optimization framework can accommodate the LVR as a cost for the LP in the objective function. 

Automated market-making has also been studied extensively in the context of prediction markets
\cite{hanson2007logarithmic,chen2010new,chen2012utility}. 
The theory of CFMMs and the dynamics around trading with CFMMs
have been studied in DeFi \cite{angeris2020improved,angeris2019analysis,angeris2021replicating,capponi2021adoption,bartoletti2021theory,bergault2022automated},
and many different DeFi applications have been deployed or proposed using different CFMM trading functions 
\cite{uniswapv2,uniswapv3,balancer,angeris2021replicatingwithoutoracles}.

%% file: preliminaries.tex
\section{Preliminaries}
\label{sec:preliminaries}


\begin{definition}[CFMM]
\label{defn:cfmm}
A CFMM trades between two assets $X$ and $Y$, and has a set of asset reserves --- $x$ units of $X$ and $y$ units of $Y$.
Its trading rule is defined by its {\it trading function} $f(\cdot, \cdot)$ such that it accepts a trade of $\Delta_X$ units of $X$ in exchange for $\Delta_Y$ units of $Y$
if and only if $f(x, y)=f(x-\Delta_X, y + \Delta_Y)$.




\end{definition}

 All of the CFMM trading functions discussed in this work have the following properties.

\begin{assumption}
\label{ass:fn_form}
A trading function $f(\cdot, \cdot): \mathbb{R}_+^2 \rightarrow \mathbb{R}$ is continuous, non-negative, increasing in both coordinates, and strictly quasi-concave. Further, it is defined only on the non-negative orthant.
\end{assumption}

The assumption that $f$ is increasing, quasi-concave, and never holds a short position in any asset (and is therefore only defined on the non-negative orthant) is standard in the literature (e.g. \cite{angeris2020improved}).
We assume strict quasi-concavity for clarity of exposition. 
%
%
%
%
%
%
The CFMM's trading function implicitly defines a marginal exchange rate (the ``spot exchange rate'')
for a trade of infinitesimal size.

\begin{definition}[Spot Exchange Rate]
\label{defn:spot}
%
At asset reserves $(x_0, y_0),$ the spot exchange rate of a CFMM with trading function $f$ is
$-\frac{\partial f}{\partial X}/\frac{\partial f}{\partial Y}$ at $(x_0, y_0)$.

When $f$ is not differentiable, the spot exchange rate is any subgradient of $f$.
When $x_0=0$, the spot exchange rate is $[-\frac{\partial f}{\partial X}/\frac{\partial f}{\partial Y}, \infty)$,
and when $y_0=0$, the spot exchange rate is $[0,-\frac{\partial f}{\partial X}/\frac{\partial f}{\partial Y}]$.

\end{definition}

These definitions directly lead to some useful observations.
	We give the proofs in Appendix \ref{apx:prelim_obs}.


\begin{observation}
\label{obs:y_fn_p}

If $f$ is strictly quasi-concave, then for any constant $K>0$ and spot exchange rate $p$, 
there is a unique point $(x,y)$ where $f(x,y)=K$ and $p$ is a spot exchange rate at 
$(x,y)$.

\end{observation}

\begin{observation}

\label{obs:y_fn_x}

Under Assumption~\ref{ass:fn_form}, for a given constant function value $K$,
 the amount of $Y$ in the CFMM reserves
uniquely specifies the amount of $X$ in the reserves, and vice versa.

\end{observation}

Observations \ref{obs:y_fn_p} and \ref{obs:y_fn_x} imply that the amounts of $X$ and $Y$ in the CFMM reserves can be written as functions $\mathcal X(p)$ and $\mathcal Y(p)$ of its spot exchange rate for the trading function equals constant $K.$

In the rest of the discussion, we describe CFMM reserve states by the amount of $Y$ in the reserves. 

\begin{observation}
\label{obs:y_nondecreasing}
$\mathcal Y(p)$ is monotone nondecreasing.
\end{observation}

%% file: model.tex
\section{Model}
\label{sec:model}

As used in Definition~\ref{defn:spot}, we adopt the notation wherein exchange rates are given as the rate of a unit of $X$ in terms of $Y$ (i.e., a trade of $x$ units of X for $y$ units of $Y$ implies an exchange rate of $p^\prime = \frac{y}{x}$). 
%
%
Unless specified otherwise, $p$ refers to the CFMM spot exchange rate, $\hat p$ denotes the exchange rate in an external market, and $p^\prime$ denotes the exchange rate of a particular trade.

We now turn to our trading model and our formulation of market liquidity.

\begin{definition}[System Model] \label{defn:system_model}
	\begin{enumerate}
		\item
			There are two assets $X$ and $Y$, and a relatively liquid ``primary'' external market that provides a (public) reference exchange rate $\hat p$
			between $X$ and $Y$.
		\item
			An LP creates a CFMM that trades between $X$ and $Y$
			 by providing an initial set of reserves and choosing a CFMM trading function.
		\item Whenever the reference exchange rate $\hat p$ on the external market changes, arbitrageurs
		immediately realign the CFMM's spot exchange rate $p$ with the reference exchange rate.
		\footnote{
			There is always a strictly profitable arbitrage trade to be made when the CFMM's spot exchange rate
			differs from the reference exchange rate \cite{angeris2021optimal}; this phenomena is akin to 
			``stale quote sniping'' in traditional exchanges \cite{baldauf2020high}.
		} \label{item:arb}
 
		\item 	At each time step, a trade request arrives with probability $q: 0<q<1$ (Definitions \ref{defn:trade} and \ref{defn:trademodel}).
	\end{enumerate}

 We assume, however, that for small fluctuations in the CFMM spot exchange rate resulting from small trades, arbitrageurs do not realign the CFMM spot exchange rate. This assumption is reasonable since the trading fee and other associated costs (e.g., gas fee in DeFi) make such an action unprofitable. 
\end{definition}


Since the reference exchange rate is public knowledge, traders using the CFMM can compare the exchange rate that a CFMM offers
with the reference rate. This difference is the \emph{slippage} of a trade.

\begin{definition}[Slippage]
\label{defn:slippage}
The exchange rate of a trade of $y$ units of $Y$ for $x$ units of $X$ is $p^\prime=y/x$. 
Relative to a reference exchange rate of $\hat p$ units of $Y$ per $X$,
the \emph{slippage} of this trade is $(p^\prime- \hat p)/ \hat p.$ 

\end{definition}

Traders in our model are willing to tolerate a fixed amount of maximum slippage $\varepsilon$.  

\begin{definition}[Trade Request]
\label{defn:trade}

A \emph{Trade Request} with a CFMM is a request to SELL or BUY $k$ units of $X$ or $Y$,
on the condition that the slippage of the trade is at most $\varepsilon$ relative to the reference exchange rate $\hat p$ --- in other words,
a trade request is a tuple
(SELL or BUY, X or Y, k, $\hat p$, $\varepsilon$). 
\end{definition}

%

\begin{definition}[Trade Success]
\label{defn:trade_success}
A trade request buying $X$ for $Y$ with maximum slippage $\varepsilon$ 
succeeds if and only if the CFMM can satisfy the entire trade with an exchange rate 
$p^\prime$ and, for the reference exchange rate $\hat p$,
$p^\prime/\hat p \leq 1+\varepsilon$. Similarly, a trade request selling $X$ for $Y$ succeeds if and only if the CFMM can satisfy the entire trade with an exchange rate 
$p^\prime$ such that
$p^\prime /\hat p \geq 1/(1+\varepsilon)$.
%
%
\end{definition}

Trade requests are not partially fulfilled.  Failed requests are not retried and are deleted. If a request succeeds, the CFMM transfers assets accordingly.
Otherwise, the CFMM's reserves are unchanged.
This notion of trade success mirrors the operation of CFMMs in practice;
users supply a trade size, exchange rate, and slippage parameter when submitting a trade request 
(e.g. \cite{uniswapinterface,balancerinterface}).

Putting these definitions together gives the trading model of our study.

\begin{definition}[Trade Model]
\label{defn:trademodel}
There exists a static (in the short term) reference exchange rate $\hat p$. 
The size of the trade request is drawn from distribution $size(\cdot)$. The choice of X or Y is arbitrary, but the trade is for BUY or SELL with equal probability.  
Each request has the same maximum slippage $\varepsilon$.

\end{definition}

Every successful trade changes the reserves of the CFMM -- this model induces a Markov chain on the state of the CFMM's reserves.  We assume that the Markov chain, at a given reference exchange rate, has sufficient time to mix before the reference exchange rate changes.  
Natural restrictions on the distribution of the trades (made explicit below) make this Markov chain ergodic.  We study, therefore,
the expected fraction of trade requests that a CFMM can satisfy when its state is drawn from the stationary distribution of this Markov chain
(we formalize this notion in Definition \ref{defn:yfail}).

\input{liquidity}
\input{inefficiency}

\input{belief}

%% file: liquidity.tex
\subsection{Liquidity}
\label{sec:liquidity}

Informally, a CFMM with high \emph{liquidity} at a given exchange rate can sell many units of $X$ before its spot exchange rate changes substantially.
Definition \ref{defn:Leps} captures precisely the set of asset reserve states of a CFMM in which the CFMM's spot exchange rate $p$ is at most a $1+\varepsilon$ 
factor away from the reference exchange rate $\hat p$.  
Recall from Observation \ref{obs:y_fn_p} that the amount of asset $Y$ in a CFMM's reserves can be expressed as a function $Y(p)$ of the CFMM spot exchange rate $p$. 

\begin{definition}[$L_\varepsilon(\hat p)$]
\label{defn:Leps}

$L_\varepsilon(\hat p)$ is the interval $\left[\mathcal Y( \frac{\hat p}{(1+\varepsilon)}) ,\mathcal Y(\hat p(1+\varepsilon)) \right]$.
%
  
\end{definition}
By Observation \ref{obs:y_nondecreasing}, $\mathcal Y( \frac{\hat p}{(1+\varepsilon)}) \leq \mathcal Y(\hat p(1+\varepsilon))$, so $L_\varepsilon(\hat p)$ is always well-defined.

 Recall that $\hat p$ is the exchange rate of a unit of $X$ in terms of $Y$.  As motivation for the choice of this definition, consider the case where $X$ is a volatile asset and $Y$ is the base numeraire currency. Here, Definition~\ref{defn:Leps} precisely captures the amount of \emph{capital} allocated to market-making in a range where the spot exchange rate of the volatile asset $X$ is within a $1+\varepsilon$ factor of its reference exchange rate.  
 In the general case where neither $X$ nor $Y$ is the base numeraire currency, the actual amount of capital (in terms of base numeraire) allocated to market-making at a certain price point $\hat p$ depends on the exchange rates of both $X$ and $Y$ in the base numeraire. 
 However, a similar intuition holds.



%
Since $\varepsilon$ is expected to be small in practice, and to facilitate easier analysis in the rest of the paper, it is useful to extend Definition \ref{defn:Leps} to study the liquidity at a single exchange rate.

\begin{definition}[Liquidity]
\label{defn:liquidity}

The liquidity at an exchange rate $\hat p$, $L(\hat p)$, is $\lim_{\theta\rightarrow 0} \frac{\vert L_\theta(\hat p)\vert}{2\ln(1+\theta)}$.

\end{definition}

Observe that $L(\hat p)$ naturally captures an allocation of capital to market-making on the full range of exchange rates. Recall from the System Model~\ref{defn:system_model} (point~\ref{item:arb}) that the arbitrageurs always realign the CFMM's spot exchange rate to the reference exchange rate. Therefore, here on, we denote the liquidity of a CFMM as a function of its spot exchange rate. 
Lemma~\ref{lemma:y_differentiable} enables a natural restatement of $L(p)$ in terms of $\mathcal Y(p)$ in Lemma~\ref{lemma:y_diffable}. We include the proof of Lemma \ref{lemma:y_differentiable}  in Appendix~\S \ref{apx:prelim_obs}. Lemma~\ref{lemma:y_diffable} follows from Definitions~\ref{defn:Leps} and~\ref{defn:liquidity}.

\newcommand{\lemmaydiffable}
{
	The function $\mathcal Y(\cdot)$ is differentiable when the trading function $f$ is twice-differentiable on the nonnegative orthant, $f$ is $0$ when $x=0$ or $y=0$,
	and Assumption \ref{ass:fn_form} holds.
}

\begin{lemma}
\label{lemma:y_differentiable}

\lemmaydiffable{}

\end{lemma}


\newcommand{\lemmaypdiffable}
{
If the function $\mathcal Y(\cdot)$ is differentiable, then $L(p)=\frac{d \mathcal Y(p)}{d\ln(p)}$.
}

\begin{lemma}
\label{lemma:y_diffable}
\lemmaypdiffable{}

\end{lemma}


The definition of liquidity implied by Lemma~\ref{lemma:y_diffable} is closely related to other definitions of liquidity in the literature.
The Uniswap V3 whitepaper \cite{uniswapv3} uses $\frac{d \mathcal Y(p)}{d\sqrt{p}}$, which is equivalent to $L(p)/\sqrt{p}$. Papers that build strategies for LPs on the Uniswap V3 protocol also adopt the same definition of liquidity \cite{neuder2021strategic,heimbach2022risks,fan2022differential} as \cite{uniswapv3}.
Milionis et al. \cite{milionis2022automated} use $\frac{-d\mathcal X(p)}{dp}$ for liquidity which is equivalent to $L(p)/p^2$; they also introduce a notion of ``instantaneous Loss-Versus-Rebalancing'' for the CFMM expected cost of operation. This quantity is proportional to $-\frac{p^2 d \mathcal X(p)}{dp}$ (Theorem 1, \cite{milionis2022automated}),
which is equivalent to $L(p)$ up to constant multipliers. 
%
%
%
%
 We conclude this subsection with some convenient facts about $L(p)$.

\newcommand{\observationliquidity}
{
	\begin{enumerate}
\item $ L(p)  =   \frac{d \mathcal Y(p)}{d~\ln(p)} = p\frac{d \mathcal X(p)}{d~\ln(1/p)}$ (when $\mathcal Y(p)$ is differentiable).
        \item 
	$ |L_\varepsilon(\hat p)| = \int_{\hat p/(1+\varepsilon)}^{\hat p(1+\varepsilon)}\frac{L(p)}{p} dp.$
 \item
	The amount of $Y$ that enters the CFMM's reserves as the spot exchange rate moves from $p_1$ to $p_2$ (for $p_1 < p_2$) is
	$\int_{p_1}^{p_2} \frac{L(p)}{p} dp.$
 
	\item
	The amount of $Y$ in a CFMM's reserves, with current spot exchange rate $p_0$, is $\mathcal Y(p_0) = \int_{0}^{p_0} \frac{L(p)}{p}dp$.

	\item
	The amount of $X$ in a CFMM's reserves, with current spot exchange rate $p_0$, is $\mathcal X(p_0) = \int_{p_0}^{\infty} \frac{L(p)}{p^2}dp$.

\end{enumerate}
}
\begin{observation}
\label{obs:alloc}
    \observationliquidity{}
\end{observation}
\label{lemma:defn_independent}
 Point 1 follows from the fact that $d \mathcal{X}(p) = p d \mathcal{Y}(p)$ and points 2-5 follow from Lemma~\ref{lemma:y_diffable}.

%% file: inefficiency.tex
\subsection{CFMM Inefficiency}
\label{sec:inefficiency}

We need an expression approximating the fraction of trade requests a CFMM fails
to satisfy.  As discussed above, the trading model given in Definition \ref{defn:trademodel} induces a Markov chain
on the state of a CFMM's reserves. We wish to quantify the expected fraction of trades
that fail during the evolution of this chain.  We assume that the reference exchange rate
changes relatively infrequently (so that this Markov chain has time to mix)
and study the chain's stationary distribution.

The precise details of the induced Markov chain (we give an example below and another in Appendix \S\ref{sec:cont}) 
depends on the trade size distribution
(and instantiation-specific assumptions).  However, common to many natural distributions
is the phenomenon that (when the reference exchange rate is $\hat p$) 
the chance that a trade request of size $k$ units of $Y$
fails is approximately $\frac{k}{L_\varepsilon(\hat p)}$. 
This approximation is closest when the sizes of the trades are much smaller than $L_\varepsilon(\hat p)$. 

\input{discrete_trades}

%% file: discrete_trades.tex
\subsubsection{Example: Constant Trade Sizes}
%
%
\begin{definition}[Size-k Trade Distribution]
\label{defn:dist_k}
At each time step, a trade request arrives with probability $q: 0<q<1$ and buys or sells $k$ units of $Y$, 
where buying or selling is chosen with equal probability (and each request tolerates a constant slippage $\varepsilon$).
\end{definition}

For the rest of this section, assume that the reference exchange rate is some unchanging $\hat p$.  The requirement that $q<1$ ensures
that the Markov chain is ergodic.

\newcommand{\lemmabounded}
{

When trades are drawn from the size-k trade distribution (Definition \ref{defn:dist_k}), if the CFMM starts with $y_0$ units of $Y$, 
then the stationary distribution of the induced Markov chain
is uniform over the points 
$\lbrace y_0+ kn~\vert ~n\in \mathbb{Z},~ n_{min} \leq n \leq n_{max}\rbrace$ for some integers $n_{min},~ n_{max}$.

Furthermore, $n_{max}-n_{min}+2 \geq \frac{\vert L_\varepsilon( \hat p) \vert}{k} \geq n_{max}-n_{min} - 2$, where $ \hat p$ is the reference exchange rate.
}

\begin{lemma}
\label{lemma:bounded}
\lemmabounded{}

\end{lemma}

\begin{proof}

The only states reachable from $y_0$ under the size-k trade distribution are a subset of the points
$\lbrace y_0 + kn ~\vert ~ n\in \mathbb{Z}\rbrace$.  
A trade of size $k$ and maximum slippage $\varepsilon$ fails if the spot exchange rate of the CFMM is already above $\hat p(1+\varepsilon)$.
Thus, there must be some $n_{max}$ such that $y_0 + kn_{max}$ upper bounds the reachable state space.  A similar argument shows that $n_{min}$ must exist.
Note that $n_{min}$ and $n_{max}$ always exist, even when $k \gg \vert L_\varepsilon(p)\vert$ (in which case $n_{min}=n_{max}=0$).

By the quasi-concavity of the CFMM trading function, the overall exchange rate of a trade must be between the spot exchange rates
before and after the trade.  Therefore, 
for any $n\in\mathbb{Z}$, a trade to sell $k$ units of $Y$ must succeed if $y_0+kn + k\in L_\varepsilon(p)$.
Thus, $n_{max}$ must be such that $y_0+(n_{max}+1)k \notin L_\varepsilon(\hat p)$ and $y_0+(n_{max}-1)k\in L_\varepsilon(p)$.
A similar argument holds for $n_{min}$.

In other words, the set of reachable states is a sequence of discrete points, all but the endpoints of which must be in
$L_\varepsilon(\hat p)$.  Thus, $(n_{max}+1)-(n_{min}-1) \geq \frac{\vert L_\varepsilon(\hat p)\vert}{k}$, 
and $(n_{max}-1) - (n_{min} + 1) \leq \frac{\vert L_\varepsilon(\hat p)\vert}{k}$.

The Markov chain, therefore, is a random walk on a finite sequence of points with an equal probability of moving in either direction (remaining in place at the endpoints instead of walking beyond the end). Standard results on Markov chains (e.g. Example 1.12, \cite{levin2017markov}) show that the Markov chain is ergodic and the stationary distribution is uniform over these points.
\extraqed{}
\end{proof}

Once this Markov chain mixes, therefore, the chance at any timestep that a trade fails is the chance that the trade
fails if the CFMM is in a randomly sampled state on this Markov chain.

\begin{lemma}
\label{lemma:approxfail}
When trades are drawn from the discrete distribution of size $k$ (Definition \ref{defn:dist_k}),
the probability that a trade fails is between $\min\left(1, \left\vert \frac{k}{\vert L_\varepsilon(\hat p)\vert - k} \right\vert \right)$ and
$\frac{k}{\vert L_\varepsilon(\hat p)\vert + k}$. 



\end{lemma}

\begin{proof}

Let $n=n_{max}-n_{min}$, as defined in Lemma \ref{lemma:bounded}.  
The trade request failure chance is $\frac{1}{n+1}$ because trade requests only fail at the endpoints of the sequence of reachable states. At $y_0 + kn_{max},$ sell requests fail and at $y_0 + kn_{min},$ buy requests fail.

Since $n +2 \geq \frac{\vert L_\varepsilon(\hat p)\vert}{k}$ and $n-2\leq \frac{\vert L_\varepsilon(\hat p)\vert}{k}$ (from Lemma~\ref{lemma:bounded}), we have that
$\frac{k}{\vert L_\varepsilon(\hat p)\vert + k}\leq \frac{1}{n+1}\leq \frac{k}{\vert L_\varepsilon(\hat p)\vert - k}$. The clipping to $1$ is required for the probability in the case where $k \geq \vert L_\varepsilon(\hat p)\vert/2.$
\extraqed{}
\end{proof}

When $k << \vert L_\varepsilon(\hat p)\vert,$ the trade failure probability is closely approximated by $\frac{k}{\vert L_\varepsilon(\hat p)\vert}$ where the approximation error is $O\left(\frac{k^2}{\vert L_\varepsilon(\hat p)\vert^2}\right).$ The rest of this work makes the following assumption.

\begin{assumption}[Small Trade Size]
\label{assumption:small_trade_size}
Trade sizes are upper bounded by a constant.
\end{assumption}

Assumption~\ref{assumption:small_trade_size} is not required to study the model in general -- traders can submit trades of size comparable to  $\vert L_\varepsilon(\hat p)\vert$ units of $Y$, and it will be successful with non-zero probability. However, the assumption enables us to approximate the trade failure probability, which is required to compile it into a metric for the CFMM designers.
See also that $|L_\varepsilon(\hat p)| = \int_{\hat p/(1+\varepsilon)}^{\hat p(1+\varepsilon)}L(p)~d\ln(p)$
(Observation \ref{obs:alloc}).  If $L(\cdot)$ is relatively constant in a neighbourhood of exchange
rate $p_0$,
then $|L_\varepsilon(p_0)| \sim L(p_0)\cdot 2\ln(1+\varepsilon)$, and so, under Assumption~\ref{assumption:small_trade_size}, for any $\varepsilon,$
 the chance that a trade fails is proportional to $k/L(p_0)$.

With this in mind, we define the following ``CFMM inefficiency'' metric.

\begin{definition}[CFMM Inefficiency]
\label{defn:yfail}
The CFMM's inefficiency at an exchange rate $\hat p$, with regard to a trade of size $k$ units of $Y$, is $\frac{k}{L(\hat p)}$.
The inefficiency of a trade denominated in $X$ is equivalent to the inefficiency of a trade of size $k \hat p$ units of $Y$.
\end{definition}

This metric has important implications for the performance of a CFMM. Consider, for example, a trader submitting a trade request of size $k$ units of $Y$ repeatedly until it succeeds. The expected number of times they have to submit the trade is $1/ (1-\frac{k}{L(\hat p)}).$ Apart from being an important metric in itself, the CFMM inefficiency is also a crucial factor when considering the LP's profits, as we will see in \S\ref{sec:net_profit}. Since the CFMM inefficiency is a convex function of each $L(\hat p),$ it can directly be incorporated in the objective function of our optimization framework in \S\ref{sec:optimization}.

%% file: belief.tex
\subsection{A Liquidity Provider's Beliefs}

We represent an LP's beliefs on future asset prices as a function of a base ``numeraire'' currency (such as USD),
instead of one of $X$ or $Y$. This is because traders and LPs usually 
denominate their profits, losses, and trade amounts in their native currency. Note, however, that this does not restrict one from studying the case where one of $X$ and $Y$ is the numeraire itself.

\begin{definition}[LP's Belief]
\label{defn:belief}
The {\it belief} of an LP is a function $\psi(\cdot,\cdot): \mathbb{R}_+^2 \rightarrow \mathbb{R}_+$
such that it believes that at a future time, asset $X$ will have price $p_X$ 
(relative to the numeraire) and $Y$ will have price $p_Y$ with probability proportional
 to $\psi(p_X, ~p_Y)$.

A belief function $\psi(\cdot,\cdot)$ has the following properties:

\begin{enumerate}
	\item 
		$\psi$ is integrable on any set of the form $\lbrace p_X, ~p_Y ~\vert~ p_1 \leq p_X/p_Y \leq p_2\rbrace$ for $p_1, ~p_2 \neq 0$.
	\item
		There exists $p_1, p_2$ so that the integral of $\psi$ on the set $\lbrace p_X, ~p_Y ~\vert~ p_1 \leq p_X/p_Y \leq p_2\rbrace$ is nonzero.
	\item 
		The set $\lbrace p_X, p_Y~\vert~\psi(p_X, p_Y) >0\rbrace$ is open, and $\psi$ is continuously differentiable on this set.

\item The integral of $\psi(p_X,~p_Y)$ over its entire domain, $N_\psi = \iint_{p_X,~p_Y} \psi(p_X, p_Y) dp_X~dp_Y,$ is a finite positive value.
\end{enumerate}


\end{definition}

We do not normalize the belief function to integrate to 1 for ease of analysis later in the paper.
This definition is strictly more flexible than a one-dimensional notion of a belief (i.e. a belief on the exchange rate between $X$ and $Y$).
A one-dimensional belief could be defined, for example, as nonzero only on the horizontal line where $p_Y=1$ (with an appropriate adjustment to the notion of integrating over the belief). 
 This flexibility will be important when we turn to the incentives of profit-seeking LPs
(\S \ref{sec:loss}).

\input{price_dynamics}

%% file: price_dynamics.tex
\subsection{Belief Functions From Price Dynamics}
\label{sec:price_dyn}

\LVR
{

An LP might not have just a belief about the distribution of future asset prices,
but also some belief about how an asset's price will evolve over time.  Applying time-discounting to beliefs about
dynamics results in a belief distribution as in Definition \ref{defn:belief}.

Let $g(p_X, p_Y)$ be any continuous, integrable function of asset prices,
and $p_X^t$ and $p_Y^t$ be stochastic processes that are believed to represent
future asset price dynamics. Let $\rho_t(p_X, p_Y)$ be the joint probability density function at time $t$ of $p_X$ and $p_Y$ induced by the stochastic processes. 
%
Denote the value of $g$ at time $t$ by $g_t.$ The expected value of $g_t$ is 
$$ \mathbb{E}(g_t) =  \iint_{p_X, p_Y}~g(p_X, p_Y) \rho_t(p_X, p_Y)~dp_X~dp_Y.$$
Denote the time-discounted value of $g$ at the initial time, with discounting parameter $\gamma$ by $g^{(\gamma)}$. By linearity of expectation, the expected value of $g^{(\gamma)}$ is
\begin{align*}
   \mathbb{E}(g^{(\gamma)}) &= \int_{t=0}^{\infty}  e^{-\gamma t} \left(\iint_{p_X, p_Y} ~g(p_X, p_Y) \rho_t(p_X, p_Y)~dp_X~dp_Y\right)~dt\\
    &= \iint_{p_X, p_Y} \int_{t=0}^{\infty}   e^{-\gamma t}~g(p_X, p_Y)~\rho_t(p_X, p_Y) ~dt~ dp_X ~dp_Y
\end{align*}
%
%
%
Observe that $\mathbb{E}(g^{(\gamma)})$, where the expectation is over the price dynamics, is therefore equivalent to the expected value of
$g$ with respect to the static belief function $\psi(p_X, p_Y)=\int_{t=0}^{\infty}   e^{-\gamma t}~\rho_t(p_X, p_Y) ~dt$.

This holds for any integrable function $g$, which includes CFMM inefficiency (as in Proposition \ref{prop:failchance}) but also
expected profit and loss (as in \S \ref{sec:net_profit}).

This framework captures the geometric Brownian motion \cite{uhlenbeck1930theory} model of price dynamics via:  

\begin{flushleft}
$\rho_t(p_X, p_Y)= \frac{1}{2 \pi}\, \frac{1}{p_X p_Y \sigma_X \sigma_Y t}\, \exp \left( -\frac{ \left( \ln p_X - \ln P_{X} - \left( \mu_X - \frac{1}{2} \sigma_X^2 \right) t \right)^2}{2\sigma_X^2 t} -\frac{ \left( \ln p_Y - \ln P_{Y} - \left( \mu_Y - \frac{1}{2} \sigma_Y^2 \right) t \right)^2}{2\sigma_Y^2 t} \right).$    
\end{flushleft}

Here, $P_X$ and $P_Y$ are the initial exchange rates of $X$ and $Y$ relative to the numeraire. $\mu_X$ and $\mu_Y$ are the drift parameters in the underlying Brownian motion of the log of $p_X$ and $p_Y.$ $\sigma^2_X$ and $\sigma^2_Y$ are the corresponding variances.   With time discounting, this induces the following belief function.
\begin{flushleft}
$\psi(p_X, p_Y)=\int\limits_{t=0}^{\infty}   e^{-\gamma t} \frac{1}{2 \pi}\, \frac{1}{p_X p_Y \sigma_X \sigma_Y t}\, \exp \Big( -\frac{ \left( \ln p_X - \ln P_{X} - \left( \mu_X - \frac{1}{2} \sigma_X^2 \right) t \right)^2}{2\sigma_X^2 t} -\frac{ \left( \ln p_Y - \ln P_{Y} - \left( \mu_Y - \frac{1}{2} \sigma_Y^2 \right) t \right)^2}{2\sigma_Y^2 t} \Big) dt.~\refstepcounter{equation}(\theequation)\label{eq:BS}$ 
\end{flushleft} 

}

%% file: optimization.tex
\section{Optimizing for Liquidity Provision}
\label{sec:optimization}

How should LPs allocate capital to market-making at different exchange rates?  
This question is the core topic of our work.
At any point in time, only the capital deployed near the reference exchange rate is useable for market-making.
Thus, the ``optimal'' CFMM design necessarily depends on an LP's belief on the distribution of future exchange rates.

We show here that an LP's beliefs on future asset valuations can be compiled into an optimal
CFMM design, which is the solution to a convex optimization problem (Theorem \ref{thm:problem}).  Specifically,
the optimization framework outputs a capital allocation $L(\cdot)$ (as in Definition \ref{defn:liquidity}) that minimizes the expected CFMM inefficiency (Proposition \ref{prop:failchance}).  
Ultimately, we show that this relationship goes both ways; a liquidity allocation uniquely specifies an equivalence class
of beliefs (Corollary \ref{cor:lp_to_equiv}). 
Per Observation \ref{obs:alloc}, a liquidity allocation $L(\cdot)$ fully
specifies a CFMM trading function.

This section discusses ``optimality'' from a viewpoint of minimizing CFMM inefficiency; however, we show in \S \ref{sec:fees} that 
this optimization framework, with a different objective function,
computes a CFMM that maximizes expected CFMM fee revenue. Furthermore,
we show in \S \ref{sec:loss} how to modify the objective of this program to account for losses incurred
during CFMM operation.

\subsection{A Convex Program for Optimal Liquidity Allocation}

\subsubsection{Objective: Minimize Expected CFMM Inefficiency}
\newcommand{\propfailchance}
{
Suppose every trade order on a CFMM is for one unit numeraire's worth of either $X$ or $Y$, and buys or sells the asset in question with equal probability.
The expected CFMM inefficiency is 
$ \frac{1}{N_\psi}\iint_{p_X, p_Y}\frac{\psi(p_X, p_Y)}{p_YL(p_X/p_Y)}dp_X~dp_Y $.
We define the integral only where $\psi(p_X, p_Y) > 0$. Further, we define $\psi(p_X, p_Y)/L(p_X/p_Y)$ to be $\infty$
when $L(p_X/p_Y) = 0$.  $N_\psi$ is as in Definition \ref{defn:belief}.
}
\begin{proposition}

\label{prop:failchance}
\propfailchance{}


\end{proposition}

\begin{proof}

Suppose that a trader order is for $1$ unit of numeraire's worth of $X$ with probability $\alpha$, 
and for $1$ unit of numeraire's worth of $Y$ with probability $1 - \alpha$. 
%
The size of a trade denominated in $X$ is therefore $1/p_X$, and the size of a trade denominated in $Y$ is
$1/p_Y$.  

Recall from Definition ~\ref{defn:yfail} that at a given set of reference prices $p_X, p_Y$, 
the CFMM inefficiency for a trade buying or selling $1$ numeraire's worth of $X$ is
$\frac{\hat p}{p_X}\frac{1}{L(p_X/p_Y)} = \frac{p_X}{p_Y p_X}\frac{1}{L(p_X/p_Y)} = \frac{1}{p_Y L(p_X/p_Y)}$.
Similarly, also from Definition ~\ref{defn:yfail},
the CFMM inefficiency corresponding to a trade of $1$ numeraire's worth of $Y$ is $\frac{1}{p_Y}\frac{1}{L(p_X/p_Y)}$.
Hence, the overall expected CFMM inefficiency is
\begin{equation*} 
\frac{1}{N_\psi}\iint_{p_X, p_Y} \psi(p_X, p_Y)\left(\frac{\alpha}{p_YL(p_X/p_Y)} + \frac{1 - \alpha}{p_YL(p_X/p_Y)}\right) dp_X~dp_Y,
\end{equation*}
%
%
\begin{equation} \label{eq:failure_chance}
=\frac{1}{N_\psi} \iint_{p_X, p_Y} \frac{\psi(p_X, p_Y)}{p_YL(p_X/p_Y)} dp_X~dp_Y. \qedhere
\end{equation} \extraqed{}
\end{proof}

For clarity of exposition, we focus on the scenario where each order trades $1$ unit of the numeraire's worth of value.  
Our model can study, however, scenarios where for general trade sizes and also when the trade size is a function of $p_X$ and $p_Y$.
The CFMM inefficiency is a linear function of trade size. A distribution of trade sizes can be multiplied with the belief function.

Proposition \ref{prop:failchance} also implies that the trade failure chance is the same for a trader buying $X$ or $Y$. The $p_Y$ in the denominator of the integrand in equation~\eqref{eq:failure_chance} appears because the liquidity $L(\cdot)$ is defined with respect to the reserves of asset Y, i.e., $\mathcal Y(\cdot)$ (recall Lemma~\ref{lemma:y_diffable}). Overall, there is no distinction between $X$ and $Y$ for the purpose of the CFMM inefficiency. \\ 


\subsubsection{Constraints: A Finite Budget for Market-Making} ~

The asset reserves of a CFMM are finite.  Clearly, the best CFMM to minimize expected inefficiency has liquidity $L(p)=\infty$ at every exchange rate $p$, 
but this would require
an infinite amount of each asset (Observation \ref{obs:alloc}).  
We model an LP with a fixed budget $B$ who creates a CFMM
when the reference exchange rates of $X$ and $Y$ in the numeraire are $P_X$ and $P_Y$, respectively.
  With this budget, the LP can purchase (or borrow)
any amount of $X$ and $Y$, say, $X_0$ and $Y_0$, subject to the constraint that $P_X X_0 + P_Y Y_0 \leq B$. 
With this intuition, we have the following technical lemmas:

\begin{lemma}
\label{lemma:supply_constraints}

Given a purchasing choice of $X_0$ and $Y_0$, the LP can choose $L(\cdot)$ and set the initial spot exchange rate of the CFMM
to be $p_0$, subject to the following asset conservation constraints.

\begin{enumerate}
	\item
		$\int_{0}^{p_0} \frac{L(p)}{p} dp\leq Y_0$

	\item
		$\int_{p_0}^{\infty} \frac{L(p)}{p^2} dp \leq X_0$
\end{enumerate}

\end{lemma}

\begin{proof}
Follows from Observation \ref{obs:alloc}.
\end{proof}

\begin{lemma}
For any two budgets $B,B^\prime$ with $B^\prime > B$ and any capital allocation $L_1(\cdot)$ satisfying the constraints of 
Lemma \ref{lemma:supply_constraints} with budget $B$, there exists a capital allocation $L_2(\cdot)$ satisfying the constraints
 of Lemma \ref{lemma:supply_constraints} 
using the larger budget $B^\prime$ that gives a strictly lower expected CFMM inefficiency.
\end{lemma}

\begin{proof}
Duplicate  $L_1(\cdot)$ and  allocate the capital $B^\prime-B$ to any $p$ with $\psi(p, 1)>0$ to build $L_2(\cdot)$. \extraqed{}
\end{proof}

A rational LP sets the initial spot exchange rate of the CFMM to be equal to the current
reference exchange rate (i.e. $p_0=\frac{P_X}{P_Y}$).  If not, a trader could arbitrage the CFMM against an external market.  The arbitrage
profit of this trader is the LP's loss, which effectively reduces the LP's initial budget.


Our convex program combines the above objective and constraints to compute an optimal liquidity allocation $L(p)$.
The core of the rest of this work is in using this
program to understand the relationship between LP beliefs and optimal liquidity allocations.

\begin{theorem}
\label{thm:problem}

Suppose that the initial reference prices of assets $X$ and $Y$ are $P_X$ and $P_Y$, and that an LP has
initial budget $B > 0$ and belief function $\psi(\cdot, \cdot)$.

The optimal liquidity provision strategy, $L(\cdot)$, is the solution to the following convex optimization problem (COP).  The decision
variables are $X_0, Y_0$, and $L(p)$ for each exchange rate $p>0.$
\footnote{
The optimization is over a Banach space with one dimension for each $p>0$; we elide this technicality when possible for clarity of exposition.
}
\begin{align}
minimize ~& \iint_{p_X, p_Y} \frac{\psi(p_X, p_Y)}{p_YL(p_X/p_Y)} dp_X~dp_Y \tag{COP}\label{eq:cop}\\
subject~to ~& \int_{0}^{p_0} \frac{L(p)}{p} dp\leq Y_0 \tag{COP$1$} \label{eq:cop1}\\
     ~& \int_{p_0}^{\infty} \frac{L(p)}{p^2} dp \leq X_0 \tag{COP$2$} \label{eq:cop2}\\
     ~& X_0 P_X + Y_0 P_Y \leq B \tag{COP$3$} \label{eq:cop3}\\
     ~& L(p) \geq 0 ~& \forall ~p>0 \tag{COP$4$} \label{eq:cop4}
\end{align}
\end{theorem}
\begin{proof}
The $L(\cdot)$ that solves \ref{eq:cop} minimizes the expected 
transaction failure chance (the expression in Proposition \ref{prop:failchance}),\footnote{The normalization term in the denominator is dropped for clarity since it doesn't change the solution of the problem.}
while satisfying the LP's budget constraint.
The objective and the constraints are integrals of convex functions and thus are convex.

This optimization problem is over a Banach space (there are uncountably many $L(p)$).  Well-established results from the theory of optimization over Banach spaces show that optimal solutions exist (Theorem 47.C, \cite{zeidler1985})
and the KKT conditions
are well defined (\S 4.14, Proposition 1, \cite{zeidler1995}).
\extraqed{}
\end{proof}

A CFMM offers only a spot exchange rate ($X$ relative to $Y$), not a spot valuation for each asset (relative to the numeraire).
In this light, we find that the objective function of \ref{eq:cop} can be rearranged to one that depends only
on ratios of valuations.
%
%
\begin{lemma}
\label{lemma:polar}
Define $r,\theta$ to be the standard polar coordinates,
with $p_X=r\cos(\theta)$ and $p_Y=r\sin(\theta)$.
\begin{dmath*}
\iint_{p_X, p_Y} \frac{\psi(p_X, p_Y)}{p_YL(p_X/p_Y)} dp_X~dp_Y
= \int_\theta \left( \frac{1}{L(\cot(\theta))\sin(\theta)} \int_r \psi(r\cos(\theta), r\sin(\theta))dr\right) d\theta
\end{dmath*}
\end{lemma}
\begin{proof}
Follows by standard algebraic manipulations ($dp_X~dp_Y = r~dr~d\theta$). \extraqed{}
\end{proof}
%
This rearrangement reveals a useful equivalence class among LP beliefs.

\newcommand{\corequivclass}
{
Any two beliefs $\psi_1, \psi_2$ give the same optimal liquidity allocations if there exists a constant $\alpha>0$ such
that for every $\theta$, 
\begin{equation*}
\int_r \psi_1(r\cos(\theta), r\sin(\theta))dr = \alpha \int_r \psi_2(r\cos(\theta), r\sin(\theta))dr
\end{equation*}
}
\begin{corollary}
\label{cor:equivclass}
\corequivclass{}
\end{corollary}

This corollary has important implications for the closed-form results we obtain in \S\ref{sec:beliefs} for commonly deployed CFMMs. The analysis of a belief defined on the square $p_X, p_Y \in (0, P_X] \times (0,P_Y]$ gives the results for all beliefs defined analogously on $p_X, p_Y \in (0, \alpha P_X] \times (0,\alpha P_Y]$ for any $\alpha > 0.$ 

\newcommand{\corvarphi}
{
Define $\varphi_\psi(\theta) = \int_r \psi(r\cos(\theta), r\sin(\theta))dr$.
Then
\begin{dmath*}
\iint_{p_X, p_Y} \frac{\psi(p_X, p_Y)}{p_YL(p_X/p_Y)} dp_X~dp_Y
= \int_p \frac{\varphi_\psi(\cot^{-1}(p))\sin(\cot^{-1}(p))}{L(p)}dp
\end{dmath*}
}

\begin{corollary}
\label{cor:varphi}
\corvarphi{}
\end{corollary}

Corollary \ref{cor:varphi} enables a straightforward construction of a feasible solution to \ref{eq:cop}.

\begin{lemma}
\label{lemma:finite_soln}
\ref{eq:cop} always has a solution with finite objective value.
\end{lemma}

\begin{corollary}
\label{corolary:phi_and_L}
On any set of nonzero measure, we cannot have $\psi(p_X, p_Y) > 0$ and $L(p_X/p_Y) = 0.$
\end{corollary}

Proofs of Corollarys~\ref{cor:equivclass} and~\ref{cor:varphi} and Lemma~\ref{lemma:finite_soln} are in the Appendix~\ref{sec:proof_equiv},~\ref{sec:proof_varphi}, and~\ref{sec:prooffinitesol} respectively. 

\subsection{Optimality Conditions}

We first give some lemmas about the structure of optimal solutions to \ref{eq:cop}.

\newcommand{\lemmabasicopt}
{
The following hold at any optimal solution.

\begin{enumerate}
	\item 
		$\int_{0}^{p_0} \frac{L(p)}{p} dp = Y_0$
	\item 
    	$\int_{p_0}^{\infty} \frac{L(p)}{p^2} dp = X_0$
    \item
    	$X_0 P_X + Y_0 P_Y = B$

\end{enumerate}	
}

\begin{lemma}
\label{lemma:basicopt}
\lemmabasicopt{}

\end{lemma}

Lemma~\ref{lemma:basicopt} says that at optimum, the constraints of \ref{eq:cop} are tight. A full proof is in Appendix~\ref{proof_basicopt}.
Using the result of Lemma \ref{lemma:basicopt}, the KKT conditions (\S 5.5.3, \cite{boyd2004convex}) of \ref{eq:cop} are the following:

\begin{lemma}[KKT Conditions]
\label{lemma:kktconds}

Let $\lambda_Y, \lambda_X$, and $\lambda_B$ be the Lagrange multipliers for \ref{eq:cop1},~\ref{eq:cop2}, and~\ref{eq:cop3} respectively.  Let $\lbrace \lambda_{L(p)} \rbrace $ 
be the Lagrange multipliers for each $L(p)\geq 0$ constraint.

When $\varphi_\psi(\cot^{-1}(p)) > 0 :$
\begin{enumerate}

	\item		For all $p$ with $p\geq p_0$, 
		$\frac{\lambda_X}{p^2} = \frac{1}{L(p)^2}\varphi_\psi(\cot^{-1}(p))\sin(\cot^{-1}(p)) + \lambda_{L(p)}$.

	\item 		For all $p$ with $p\leq p_0$, 
		$\frac{\lambda_Y}{p} = \frac{1}{L(p)^2}\varphi_\psi(\cot^{-1}(p))\sin(\cot^{-1}(p)) + \lambda_{L(p)}$.

	\item $\lambda_X = P_X \lambda_B$ and $\lambda_Y = P_Y \lambda_B$.
\end{enumerate}

When $\varphi_\psi(\cot^{-1}(p)) = 0 :$
\begin{enumerate}

	\item

		For all $p$ with $p\geq p_0$, 
		$\frac{\lambda_X}{p^2} = \lambda_{L(p)}$.

	\item 
		For all $p$ with $p\leq p_0$, 
		$\frac{\lambda_Y}{p} = \lambda_{L(p)}$.

	\item 
		$\lambda_X = P_X \lambda_B$ and $\lambda_Y = P_Y \lambda_B$.
\end{enumerate}

\end{lemma}

\begin{proof}
These are the KKT conditions of \ref{eq:cop}.
$\lbrace L(p) \rbrace$
is a functional over a Banach space.  This functional exists for every optimal solution
by Proposition 1 of \S 4.14 of \cite{zeidler1995}.  Note that that proposition requires the objective
to be continuously differentiable in a neighbourhood of the optimal solution; this does not hold
when the optimization problem is as written and there is some $p$ so that $\varphi_\psi(\cot^{-1}(p))$ goes continuously to $0$ at $p$ (but is nonzero 
near $p$).  In this case, one could replace $L(p)$ by $L(p)+\varepsilon$ in the denominator of the objective, for some arbitrarily small $\varepsilon$.
This would cause a small distortion in $L(p)$.  We elide this technicality for clarity of exposition.  
Continuous differentiability of the objective on a neighbourhood where $L(p)>0$ for all $p$ with $\varphi_\psi(\cot^{-1}(p))>0$ follows
from the assumption that $\psi$ is continuously differentiable on the set where $\psi(p_X, p_Y)>0$, and that this set is open (in Definition~\ref{defn:belief}).
\extraqed{}
\end{proof}

\newcommand{\coryhatdefined}
{
The integral $\mathcal Y(\tilde{p})=\int_0^{\tilde{p}} \frac{L(p)dp}{p}$ is well defined for every $\tilde{p}$ 
and $\mathcal Y(\cdot)$ is monotone nondecreasing and continuous.
	
}

\begin{corollary}
\label{cor:yhat}
\coryhatdefined{}

\end{corollary}

A proof is given in Appendix~\ref{sec:proof_int_y}. Lemma \ref{lemma:kktconds} and Corollary \ref{cor:yhat} together
imply that the behaviour of a CFMM that results from an optimal solution of \ref{eq:cop} is well-defined. \\

\subsubsection{Consequences of KKT Conditions} ~

The KKT conditions immediately imply the following facts about any optimal solution of \ref{eq:cop}.

\newcommand{\lemmakktobs}
{
\begin{enumerate}
	\item
		$\lambda_Y Y_0 = \int_0^{p_0} \frac{\varphi_\psi(\cot^{-1}(p))\sin(\cot^{-1}(p))}{L(p)} dp$
		and $\lambda_X X_0 = \int_{p_0}^\infty \frac{\varphi_\psi(\cot^{-1}(p))\sin(\cot^{-1}(p))}{L(p)} dp$.

	\item

		$Y_0>0$ implies $\lambda_Y>0$.  Similarly, $X_0 > 0$ implies $\lambda_X > 0$.

	\item
		$L(p)\neq 0$ if and only if $\lambda_{L(p)}=0$ (unless, for $p\leq p_0$, $\lambda_Y=0$ or for $p\geq p_0$, $\lambda_X=0$).

	\item
		The objective value is $\lambda_Y Y_0 + \lambda_X X_0$.

	\item 
		$\frac{\lambda_X}{P_X} = \frac{\lambda_Y}{P_Y}$.

\end{enumerate}

}

\begin{lemma}
\label{lemma:kktobs}
\lemmakktobs{}

\end{lemma}

\begin{proof}

\begin{enumerate}

	\item 
		Multiply each side of the first KKT condition in Lemma \ref{lemma:kktconds} by $L(p)$ (for $p$ with nonzero $\varphi_\psi(\cot^{-1}(p))$
		to get
		$\frac{\lambda_X L(p)}{p^2} = \frac{1}{L(p)}\varphi_\psi(\cot^{-1}(p))\sin(\cot^{-1}(p)))$, integrate from $p_0$ to $\infty$,
		and apply the second item of Lemma \ref{lemma:basicopt}.
		
		A similar argument (integrating from $0$ to $p_0$) gives the expression on $\lambda_Y Y_0$.
	\item

		If $Y_0>0$, then the right side of the equation in the previous part is nonzero, so $\lambda_Y$ must be nonzero.
		The case of $\lambda_X$ is identical.

	\item
		Follows from points 1 and 2 of Lemma \ref{lemma:kktconds}.

	\item

		The right sides of the equations in the first statement add up to the objective.

	\item
		Follows from point 3 of Lemma \ref{lemma:kktconds} \qedhere
\end{enumerate}
\extraqed{}
\end{proof}

Lemma \ref{lemma:kktobs} shows that the fraction of liquidity allocated to an exchange rate $p$ is a function only of the LP's
(relative) belief that the future exchange rate will be $p$.  Specifically, except through an overall scalar,
there is no interaction between the values of $L(\cdot)$ at different relative exchange rates.

\begin{proposition}
\label{prop:lp_dependencies}

At an optimum,
$L(p)$ is a function of $\lambda_X, \lambda_Y, \varphi_\psi(\cot^{-1}(p))\sin(\cot^{-1}(p))$, and $p$.

\end{proposition}

\begin{proof}
Follows from Lemma \ref{lemma:kktconds}. \extraqed{}
\end{proof}

Proposition \ref{prop:lp_dependencies} gives several important consequences.  First, it shows that an optimal liquidity allocation
can be inverted to give a set of belief functions that lead to that liquidity allocation.


\newcommand{\corlpequiv}
{
A liquidity allocation $L(\cdot)$ and an initial spot exchange rate $p_0$ are sufficient to uniquely specify 
an equivalence class of beliefs (as defined in Corollary \ref{cor:equivclass}) for which $L(\cdot)$ is optimal.
	
}

\begin{corollary}
\label{cor:lp_to_equiv}

\corlpequiv{}
\end{corollary}

Second, Proposition \ref{prop:lp_dependencies} actually enables an explicit construction of a belief
that leads to $L(\cdot)$.


\newcommand{\corinverse}
{
Let $P_X$ and $P_Y$ be initial reference valuations, and let $L(\cdot)$ denote a liquidity allocation.
Define the belief $\psi(p_X, p_Y)$ to be $\frac{(L(p_X/p_Y))^2}{p_X/p_Y}$ when $p_X\in (0, P_X]$ and
$p_Y\in (0, P_Y]$, and to be $0$ otherwise.
Then $L(\cdot)$ is the optimal allocation for $\psi(\cdot, \cdot)$.
}
\begin{corollary}
\label{cor:inverse_problem}
\corinverse{}

\end{corollary}

Finally, the KKT conditions (Lemma \ref{lemma:kktconds}) imply that linear combinations of beliefs result in predictable
combinations of liquidity allocations. Towards this, we have the following result, which will also be useful in further proofs. Proofs of Corollaries~\ref{cor:lp_to_equiv},~\ref{cor:inverse_problem}, and~\ref{cor:addbeliefs} are in Appendix~\S\ref{sec:corr_proofs}.

\newcommand{\thmaddbeliefs}
{
Let $\psi_1, \psi_2$ be any two belief functions (that give $\varphi_{\psi_1}$ and $\varphi_{\psi_2}$)
with optimal allocations $L_1(\cdot)$ and $L_2(\cdot)$, and let $L(\cdot)$ be the optimal allocation for
$\psi_1+\psi_2$.
Then $L^2(\cdot)$ is a linear combination of $L_1^2(\cdot)$ and $L_2^2(\cdot)$.

Further, when $\varphi_{\psi_1}$ and $\varphi_{\psi_2}$ have disjoint support,
$L(\cdot)$ is a linear combination of $L_1(\cdot)$ and $L_2(\cdot)$.
	
}

\begin{corollary}

\label{cor:addbeliefs}

\thmaddbeliefs{}

\end{corollary}

%% file: beliefs.tex
\section{Common CFMMs and Beliefs}
\label{sec:beliefs}

We turn now to the CFMMs deployed in practice.  What do the choices of trading functions in large CFMMs reveal about practitioners' beliefs 
about future asset prices?  In fact, the optimal beliefs for several widely-used trading functions closely match 
the widespread but informal intuition about these systems. Recall that $P_X$ and $P_Y$ are the initial reference exchange rates. Also, recall the assumption that all trades are for the worth of 1 unit of the base numeraire currency (Proposition~\ref{prop:failchance}).

\input{uniswap_v2}

\input{uniswap_v3}

\input{balancer}

\input{lmsr}

%% file: uniswap_v2.tex
\subsection{The Uniform, Independent Belief: Constant Product Market Makers}
\label{sec:uniswapv2}

\begin{proposition}
\label{prop:uniswapv2}

Let $\psi(p_X, p_Y) =1 $ on $(0, P_X] \times (0, P_Y]$ and $0$ otherwise.
%
The liquidity allocation $L(\cdot)$ that minimizes the CFMM inefficiency is the allocation implied
by the trading function $f(x,y) = xy$.

\end{proposition}

Of course, by Corollary \ref{cor:equivclass}, the belief that gives the constant product market maker is not unique.  Importantly, 
rescaling the belief to one defined analogously on the rectangle $(0, \alpha P_X]\times (0, \alpha P_Y]$ 
for any constant $\alpha >0$ does not change the optimal liquidity allocation. This invariance of the optimal liquidity allocation to such transformations of the belief applies to all results in this section.

\begin{proof}

For the CFMM $f(x,y) = xy,$ we have $\mathcal X(p) \mathcal Y(p)=X_0Y_0$ and $p= \mathcal Y(p)/ \mathcal X(p)$. This implies $ \mathcal Y(p)^2 = pX_0Y_0$ and $\mathcal Y(p)=\sqrt{pX_0Y_0}$.  
Recall that $L(p) = \frac{d \mathcal Y(p)}{d \ln(p)}.$
This gives $L(p)=\frac{\sqrt{pX_0Y_0}}{2}$.

Corollary \ref{cor:inverse_problem} shows that a belief that leads to this liquidity allocation is
$\frac{X_0Y_0}{4}$ on the rectangle $(0, P_X] \times (0, P_Y]$ and $0$ elsewhere.  The result follows by rescaling
the belief (Corollary \ref{cor:equivclass}).
\extraqed{}
\end{proof}

Proposition \ref{prop:uniswapv2} captures the folklore intuition within Decentralized Finance regarding the circumstances
in which constant product market makers are optimal.
If an LP has no information regarding correlations in the reference 
valuations of assets (in terms of the numeraire),
 then the LP 
should choose one of these CFMMs
because it
allocates liquidity evenly across the entire range of exchange rates.
To be specific, the liquidity available at a given exchange rate for purchasing $Y$ from the CFMM is always
 proportional to the amount of $Y$ in the reserves at that exchange rate.

%% file: uniswap_v3.tex
\subsection{Uniform Beliefs on Exchange Rate Ranges: Concentrated Liquidity Positions}
\label{sec:uniswapv3}

Some CFMMs (e.g. \cite{uniswapv3}) allow LPs to create piecewise-defined trading strategies,
often called ``concentrated liquidity CFMMs.''  

\begin{definition}[Concentrated Liquidity CFMM]
A trading function $f'(x,y)= f(x+\hat{x},y+\hat{y})$
for some constants $\hat{x} > 0, \hat{y} > 0$, has nonzero liquidity on a smaller range of exchange rates than $f$.
\end{definition}

While a ``concentrated liquidity CFMM'' $f'$ can be designed with any $f,$ we focus on those which use $f(x,y) = xy$ since these have been 
widely adopted in practice after being introduced by \cite{uniswapv3}. 
 
Observe that this trading function differs from the constant product trading rule when $x$ or $y$ reaches $0$.
There exists a range of exchange rates $(p_{\min}, p_{\max})$ on which this CFMM makes the same trades as one based on the constant product rule.
Outside of this range, the CFMM makes no trades.  
There is a direct mapping from $(\hat{x}, \hat{y})$ to $(p_{\min}, p_{\max}), $ we omit it here for clarity of exposition.


This trading function corresponds to a belief pattern that is restricted in a similar way;
on the specified range of exchange rates, the belief is the same as that of Proposition \ref{prop:uniswapv2},
and $0$ otherwise.

\newcommand{\univthreethm}
{
Let $p_{min} < p_{max}$ be two arbitrary exchange rates,
and let $\psi(p_X, p_Y) = 1$ if and only if 
$0\leq p_X \leq P_X$, $0 \leq p_Y \leq P_Y$, and $p_{min} \leq p_X / p_Y \leq p_{max}$, and $0$ otherwise.
%
%
The allocation $L(\cdot)$ that maximizes the fraction of successful trades is the allocation implied
by a concentrated liquidity position with price range defined by $p_{min}$ and $p_{max}$.
}
\begin{proposition}
\label{prop:uniswapv3}

\univthreethm{}

\end{proposition}
A proof is given in Appendix~\S\ref{proof_v3}.
LPs who make markets using concentrated liquidity CFMMs implicitly expect that while the valuations of the two assets
may move up and down, their movements are correlated; that is, the exchange rate always stays within 
some range.  This belief exactly matches that intuition.  
\begin{corollary} \label{corr:comb_liq}
The belief corresponding to multiple concentrated liquidity positions, which are defined on disjoint ranges of exchange rates,
is a linear combination of the beliefs that correspond to the individual concentrated 
liquidity positions (as specified in Proposition \ref{prop:uniswapv3}).
%
%
%
\end{corollary}
\begin{proof}
Application of Corollary \ref{cor:addbeliefs}. \extraqed{}
\end{proof}

%% file: balancer.tex
\subsection{Skew in Belief Function: Weighted Product Market Makers}
\label{sec:balancer}

Weighted Constant Product Market Makers add weights to a constant product curve to get
trading functions of the form $f(x,y)=x^\alpha y$, for some constant $\alpha>0$.

\newcommand{\balancerthm}
{
%
The belief function $\psi(p_X, p_Y)=\left(\frac{p_X}{p_Y}\right)^{\frac{\alpha - 1}{\alpha + 1}}$ when 
$(p_X, p_Y) \in (0, P_X] \times (0, P_Y]$ and $0$ otherwise corresponds to the weighted product market
maker $f(x,y)=x^\alpha y$.
}

\begin{proposition}
\label{prop:balancer}

\balancerthm{}

\end{proposition}

A proof is given in Appendix~\ref{sec:expected_future_value}. This proposition shows that LPs who use weighted product market makers
 expect that the value of one asset will typically be much higher than that of the other.  
Informally, 
more liquidity is allocated towards the higher ranges of exchange rates than the lower ranges when $\alpha> 1$
(and vice versa for $\alpha < 1$).  This CFMM, therefore, can satisfy a higher fraction of trades when the exchange rate is high 
(resp. low).
On the other hand, a skewed allocation means that some price ranges suffer from much higher slippage for a fixed-size trade.
This intuition mirrors the description of how LPs
should choose the weight $\alpha$ \cite{balancerintuition}
in a public deployment \cite{balancer} of weighted product market makers.
This theorem reduces to Proposition \ref{prop:uniswapv2} when $\alpha=1$.

%% file: lmsr.tex
\subsection{Logarithmic Market Scoring Rule}
\label{sec:lmsr}

The logarithmic market scoring rule (LMSR) \cite{hanson2007logarithmic}, which has been 
used extensively in the context of prediction markets,
 corresponds to a CFMM
with trading function $f(x,y)=2 -e^{-x} - e^{-y}$ \cite{univ3paradigm}.

\newcommand{\lmsrthm}
{

The optimal trading function to minimize the expected CFMM inefficiency for the belief $\psi(p_X, p_Y)=\frac{p_Xp_Y}{(p_X+p_Y)^2}$
 when $(p_X, p_Y)\in (0, P_X]\times (0, P_Y]$ and $0$ otherwise,
 is $f(x,y)=2 -e^{-x} - e^{-y}$.

}

\begin{proposition}
\label{prop:lmsr}
\lmsrthm{}

\end{proposition}

A proof is given in Appendix~\ref{proof_lmsr}. Observe that $\psi$ is symmetric about the line $p_Y=p_X$ --- that is to say, $\psi(p_X, p_Y)=\psi(p_Y, p_X)$.
 We analyse this belief function in
polar coordinates to get a better intuition. Note
the term $\frac{r\cos(\theta)r\sin(\theta)}{(r\cos(\theta) + r\sin(\theta))^2}= \frac{\sin(2\theta)}{2(1+\sin(2\theta))}$, which is maximized at $\theta=\pi/4$.
For initial exchange rates with $P_X=P_Y$, 
the LMSR expects the relative exchange rate to concentrate about $\frac{p_X}{p_Y}=1$.  
At extreme exchange rates ($\frac{p_X}{p_Y}\rightarrow 0$ or
$\frac{p_X}{p_Y}\rightarrow\infty$), the belief goes to $0$.  
The LMSR-based CFMM correspondingly allocates very little liquidity at extreme exchange rates.

%% file: overall_profit_and_loss.tex
\section{Market-Maker Profit and Loss}
\label{sec:net_profit}

Deploying assets within a CFMM has a cost,
and LPs naturally want to understand the financial tradeoffs involved (instead of just minimising CFMM inefficiency).
Specifically, LPs in CFMMs trade off revenue from transaction fees against losses due to adverse selection.

\input{fees}

\input{loss}

\input{net_profit}

\input{bid_ask}

%% file: fees.tex
\subsection{Fee Revenue}
\label{sec:fees}

Many CFMM deployments charge a fixed-rate fee on every transaction.  
While some early instantiations automatically reinvested fee revenue within the CFMM reserves \cite{uniswapv2}, 
more recent deployments choose not to \cite{uniswapv3}.  We consider the case where fees are not automatically reinvested and,
for simplicity, the case where fees are immediately converted into the numeraire currency.


Observe that a percentage-based fee on a trade is, from the trader's point of view, equivalent to a
 multiplicative factor in the exchange rate.
Slippage also measures the deviation of a trader's received exchange rate from the reference exchange rate.
A transaction fee, therefore, has the same effect in our model on a CFMM's ability to settle trades as a
reduction in a user's tolerated slippage.


The fee revenue depends on the rate of trade requests and the fraction of trades the CFMM can settle. Recall the probability $q$ of a trade request arriving at a time step from the system model in Definition~\ref{defn:system_model} and the trade size distribution $size(\cdot)$
 from the trade model in Definition~\ref{defn:trademodel}. These can be compiled into a belief on the ``rate'' of trade requests that the LP expects the CFMM will see.

\begin{definition}[Transaction Rate Model]
\label{defn:rate}
Denote as $rate_{\delta}(p_X, p_Y)$ denote an LP's prediction on the expected volume of trades 
(in terms of the numeraire) attempted on the CFMM when the reference prices are $p_X$ and $p_Y$ and the trading fee is set to $\delta$.
\end{definition}

\newcommand{\proprevenue}
{
The expected revenue of a CFMM in one unit of time, using transaction fee $\delta$, when the mean trade-size is worth $s$ units of the numeraire, is
\begin{equation*}
\mu(rate,\psi,\delta,s) = \frac{\delta}{N_\psi} \iint_{p_X, p_Y} rate_{\delta}(p_X, p_Y) \psi(p_X, p_Y)\left(1-\frac{s}{p_YL(p_X/p_Y)}\right) dp_X~dp_Y.
\end{equation*}
}
\begin{proposition}
\label{prop:exp_revenue}

\proprevenue{}
\end{proposition}
\begin{proof}
The expected revenue of the CFMM is exactly the transaction fee times the volume of trading that goes through the CFMM, which
is equal to the predicted input trade volume, less the number of trades that the CFMM cannot settle. \extraqed{}
\end{proof}

In any real-world setting, the fee $\delta$ influences the predilection of traders to use the CFMM.
However, the transaction fee is a predetermined constant in many of the most widely used CFMM deployments.  Uniswap V2, for example, sets
a fixed $0.3\%$ fee \cite{uniswapv2}, and Uniswap V3 lets LPs choose between three choices of fee rates \cite{uniswapv3}.
In this model, LPs would consider each fee rate they are allowed to set (or perform a grid search over many fee schedules),
and then predict (through some external knowledge, outside the scope of this work) the transaction rate at that fee schedule.  

\begin{corollary}
\label{cor:fee_maximize}

The allocation $L(\cdot)$ that maximizes fee revenue is same as the $L(\cdot)$ that minimizes
\begin{equation*}
\iint_{p_X, p_Y} \frac{rate_{\delta}(p_X, p_Y)\psi(p_X, p_Y)}{p_YL(p_X/p_Y)} dP_X~dP_Y
\end{equation*}
which is in turn equal to the liquidity allocation $L(\cdot)$ which is optimal for an LP with belief
$rate_{\delta}(p_X, p_Y) * \psi(p_X, p_Y)$ and is concerned only with minimizing CFMM inefficiency.

\end{corollary}

In simpler terms, an expected distribution on future transaction rates is equivalent, in the eyes of the optimization framework,
to a belief on future prices.  A revenue-maximizing LP can therefore use the same optimization problem
(with an adjusted input) as in Theorem \ref{thm:problem}.

Definition \ref{defn:rate} assumes traders are attempting trades for exogenous reasons. However, in the real world, trade volume might depend on how effectively a market-maker provides liquidity.
Note that when the trade input rate is independent of $p_X$ and $p_Y,$ then the objective of minimizing CFMM inefficiency produces the same
 liquidity allocation as an objective of maximizing the fee collected.


%% file: loss.tex
\subsection{Liquidity Provider Loss}
\label{sec:loss}

LPs may also suffer losses when asset prices change.
As discussed in Milionis et al. \cite{milionis2022automated} and Cartea~\cite{cartea2022decentralised},
these losses come from two sources: first, from the asset price movements directly (exposure to market-risk), 
and second, from shifts in the relative exchange rate between the assets.
As the relative exchange rate on the external market shifts, arbitrageurs can trade with the CFMM to realign the CFMM's spot exchange rate
with the external market's exchange rate, making a profit in the process at the CFMM's expense. Milionis et al. \cite{milionis2022automated} show that the expected lose due to arbitrage is higher for CFMMs trading more volatile assets.

The CFMM's loss can only be defined relative to a counterfactual choice; instead of engaging in market-making, an LP could deploy
their capital in some other manner.
For example, a simple counterfactual would be to hold a fixed amount of the numeraire or a fixed quantity of each asset and, at a future time, compare the value of this strategy with the value of the CFMM's asset reserves. 

\newcommand{\thmdivergence}
{
	The expected future value of the CFMM's reserves, as per belief $\psi(p_X, p_Y)$, is
\begin{align*}
&\nu(\psi) = \frac{1}{N_\psi} \iint\limits_{p_X, p_Y}\psi(p_X, p_Y)\left({p_X \mathcal X(p_X/p_Y) + p_Y \mathcal Y(p_X/p_Y)}\right)dp_X~dp_Y \\
& = \frac{1}{N_\psi} \int\limits_{0}^\infty \left( \frac{L(p)}{p^2} \iint\limits_{p_X, p_Y}p_X\psi(p_X, p_Y)\mathds{1}_{\lbrace \frac{p_X}{p_Y} \leq p\rbrace}dp_X~dp_Y
+ \frac{L(p)}{p} \iint\limits_{p_X, p_Y}p_Y\psi(p_X, p_Y)\mathds{1}_{\lbrace \frac{p_X}{p_Y} \geq p\rbrace}dp_X~dp_Y\right)dp
\end{align*}

where $\mathcal X(p)$ and $\mathcal Y(p)$ denote the amounts of $X$ and $Y$ held in the reserves at spot exchange rate $p$, and $\mathds{1}_{E}$ is the characteristic function of the event $E$. 

This expression for $\nu(\psi)$ is a linear function of each $L(p)$.
}
\begin{proposition}
\label{prop:divergence}

\thmdivergence{}

\end{proposition}

Proposition~\ref{prop:divergence} follows from representing $\mathcal X(p_X/p_Y)$ and $\mathcal Y(p_X/p_Y)$ in terms of integration of $L(p)$ per Observation~\ref{obs:alloc} and then changing the order of integration. A full proof is in Appendix~\ref{sec:expected_future_value}.

\subsubsection{Divergence Loss} 
\label{sec:divloss}~ 

Observe that the expected value of a ``buy-and-hold" counterfactual strategy does not depend on a chosen allocation $L(\cdot)$.
In this case, the expected loss of the market maker (relative to a counterfactual strategy with expected payoff $C$) is $C - \nu(\psi)$ and is linear in each $L(p)$.  This measurement of loss is typically called ``divergence loss''.

\subsubsection{Loss-Vs-Rebalancing}
\label{sec:lvr}~

Alternatively, Milionis et al. \cite{milionis2022automated} propose the ``LVR'' metric that compares a CFMM's performance against that 
of a counterfactual strategy
that makes the same trades in asset $X$ as a CFMM but at the reference exchange rate, not at the CFMM's spot rate.  
This LVR is an expectation over a model of price dynamics but depends only on that model and the so-called ``instantaneous LVR'', which
is proportional to $-p^2\frac{ \mathcal X(p)}{dp} = L(p)$ (Theorem 1 and Lemma 1, \cite{milionis2022automated}). Milionis et al. \cite{milionis2022automated} show that the ``instantaneous LVR'' is equivalent to the loss to the CFMM due to arbitrage.

The overall LVR can be given by integrating the ``instantaneous LVR'' over the time period for which the LP operates the CFMM.
Since the overall LVR is, therefore,
a linear function of $L(\cdot)$, it can be incorporated directly into our optimization framework.

%% file: net_profit.tex
\subsection{Net Profit}

Adding capital to a CFMM is profitable in expectation if and only if the expected fee revenue
outweighs the expected loss. In the following, we analyze the net expected profit of an LP.

\begin{lemma}
\label{lemma:net_profit}

Let $\mu(\cdot)$ be as defined in Proposition \ref{prop:exp_revenue} and $\Gamma(\cdot)$ be  loss function which is linear in $L(p)$ for each $p$. 
An LP's net expected profit,  
$\mu(rate, \psi,\delta,s)-\Gamma(L(\cdot))$, is a concave function 
in $L(p)$ for each $p$.


\end{lemma}

\begin{proof}

$\mu(\cdot)$ is clearly concave, and for any loss function which is linear in each variable $L(p)$, the whole function is concave in each $L(p)$.
\extraqed{}
\end{proof}
Crucially, Lemma~\ref{lemma:net_profit} is applicable to both the divergence loss and the Loss-vs-Rebalancing.

\begin{observation}
A profit-maximizing liquidity allocation
by maximizing $\mu(rate, \psi,\delta,s)-\Gamma(L(\cdot))$, and is computable via
a convex optimization problem. 

The objective is the expression of Lemma \ref{lemma:net_profit}, and the constraints
are the same as in Theorem \ref{thm:problem}.

\end{observation}

Not only, therefore, does our optimization framework
allow an LP to design a CFMM that maximizes successful trading activity, 
but it can also guide
a profit-maximizing LP to a profit-maximizing CFMM.
Observe that the objective value of this convex program is positive precisely
when the CFMM generates a profit. 


%% file: bid_ask.tex
\subsection{Divergence Loss Shifts Liquidity Away From the Current Exchange Rate}
\label{sec:divloss_extreme}

The threat of divergence loss is in inherent tension with the potential of a CFMM to earn fee revenue by settling trades.
A CFMM can, of course, vacuously eliminate divergence loss by refusing to make any trades (by allocating no liquidity to any price),
and in general, higher liquidity at a given range of exchange rates leads to better CFMM performance at that range and higher divergence loss
if the reference exchange rate moves away from that range.


When divergence loss is accounted for in the objective, our optimization framework 
computes a liquidity allocation that trades off fee revenue with divergence loss by shifting liquidity away from the current exchange rate i.e., towards more extreme exchange rates (closer to $p=0$ and $p = \infty$). 

\LVR{
This qualitative behaviour is, however, not true of every possible loss function that an LP might add to the objective.
Each function will have its own effect; nevertheless, the optimization problem can always be solved to compute an optimal allocation.
}

\newcommand{\thmdivloss}
{
	
Let $L_1(p)$ be the optimal liquidity allocation that maximizes fee revenue ---
the solution to the optimization problem for the objective of minimizing the following:
\begin{equation*}
-\frac{\delta}{N_\psi} \iint_{p_X, p_Y} rate_{\delta}(p_X, p_Y)\psi(p_X, p_Y) \left(1-\frac{s}{p_YL(p_X/p_Y)}\right) dp_X~dp_Y
\end{equation*}

Let $L_2(p)$ be the optimal liquidity allocation that maximizes fee revenue while accounting for divergence loss ---
the solution to the optimization problem for the objective of minimizing the following:
\begin{dmath*}
- \nu(\psi) - \frac{\delta}{N_\psi} \iint_{p_X, p_Y} rate_{\delta}(p_X, p_Y)\psi(p_X, p_Y)\left(1-\frac{s}{p_YL(p_X/p_Y)}\right) dp_X~dp_Y 
\end{dmath*}

Let $ X_1=\int_{p_0}^\infty \frac{L_1(p)}{p^2}dp$ and $  X_2=\int_{p_0}^\infty \frac{L_2(p)}{p^2}dp$ be the optimal initial quantities of $X$
 for the above two problems respectively.

Then there exists some $p_1>p_0$ such that for $p_0 \leq p \leq p_1$, $\frac{L_1(p)}{X_1} \geq \frac{L_2(p)}{X_2}$ and for $p\geq p_1$,
$\frac{L_1(p)}{X_1} \leq \frac{L_2(p)}{X_2}$.  An analogous statement holds for the allocations of $Y$.

}

\begin{theorem}
\label{thm:qual_divloss}

\thmdivloss{}


\end{theorem}

The proof is technical and is given in Appendix~\ref{sec:expected_future_value}. In simple terms, the divergence loss might change the optimal initial choice of reserves, but given that choice, divergence loss
causes liquidity to shift away from the current exchange rate.  Qualitatively, the higher the magnitude of the expected divergence loss
relative to the expected fee revenue, the larger the magnitude of this effect.  This reflects the well-known intuition that 
when the reference exchange rate changes substantially,
highly concentrated liquidity positions (like in \S \ref{sec:uniswapv3}) suffer higher losses than more evenly distributed liquidity allocations (like
those of \S \ref{sec:uniswapv2}).

%% file: conclusion.tex
\section{Conclusion}

We develop in this work a convex program that, for any set of beliefs about future asset valuations, outputs a trading function
that maximizes the expected fraction of trade requests the CFMM can settle.  Careful analysis of this program allows for study of the inverse relationship
as well; for any trading function, our program can construct a class of beliefs for which the trading function is optimal.  Constructing this program
requires a new notion of the liquidity of a CFMM, 
and the core technical challenge
involves careful analysis of the KKT conditions of this program.

Unlike prior work, this program is able to explain the diversity of CFMM trading curves observed in practice.
  We analyze several CFMMs that are widely deployed in
the modern DeFi ecosystem, and show that the beliefs revealed by our model match the informal intuition of practitioners.

Furthermore, our program is versatile enough to cover the case of a profit-seeking LP that must trade off expected revenue from trading fees against loss due to arbitrage.
This program therefore can serve as a guide for practitioners when choosing a liquidity allocation in a real-world CFMM.

%% file: cost_fn_apx.tex
\section{CFMMs and Market Scoring Rules}
\label{apx:cost_fn}

We highlight here for completeness the equivalence between market scoring rules \cite{hanson2003combinatorial}
and CFMMs. Chen and Pennock \cite{chen2012utility} show that every prediction market, based on a market scoring rule,
can be represented using some ``cost function.''

A prediction market trades $n$ types of shares, each of which pays out 1 unit of a numeraire if a particular future event occurs.
The cost function $C(q)$ of \cite{chen2012utility} is a map from the total number of issued shares of each event, $q\in\mathbb{R}^n$, to some number of units of
the numeraire.  To make a trade $\delta\in\mathbb{R}^n$ with the prediction market (i.e. to change the total number of issued shares to $q+\delta$),
a user pays $z=C(q+\delta) - C(q)$ units of the numeraire to the market.

One discrepancy is that traditional formulations of prediction markets (e.g. \cite{chen2012utility,hanson2007logarithmic})
 allow an arbitrary number of shares to be issued by the market maker, but the CFMMs described in this work trade in assets with finite supplies. 
%
Suppose for the moment, however, that a CFMM could possess a negative quantity of shares (with the trading function $f$ defined
on the entirety of $\mathbb{R}^n$, instead of just the positive orthant).
This formulation of a prediction market directly gives a CFMM that trades the $n$ shares and the numeraire, with trading function 
$f(r,z)=-C(-r)+z$ for $r\in \mathbb{R}^n$ the number of shares
owned by the CFMM, and $z$ the number of units of the numeraire owned by the CFMM.  Observe that for any trade $\delta$ and 
$dz=C(-(r + \delta))-C(-r)$, $f(r, z)=f(r+\delta, z+dz)$.  This establishes the correspondence between prediction markets and CFMMs.

In our examples with the LMSR, we consider a CFMM for which $z = 0$ 
(i.e., it doesn't exchange shares for dollars, but only shares of one future event for shares of another future event). 
The cost function $C(r)$ for the LMSR is $\log(\sum_{i=1}^n \exp(-r_i)).$ The CFMM representation with this cost function follows by setting it to a constant.

%% file: continuous_trading_main_paper.tex


\begin{definition}
\label{defn:dist_cont}
Let $size(\cdot)$ be some distribution on $\mathbb{R}_{\geq 0}$ with support in a neighborhood of $0$.

A trader appears at every timestep.  The trade has size $k$ units of $Y$, where $k$ is drawn from $size(\cdot)$.
  A trade buys or sells from the CFMM with
equal probability.
\end{definition}

This definition implicitly encodes an assumption that the amount of trading from $X$ to $Y$ is balanced in expectation
against the amount of trading from $Y$ to $X$.

An additional assumption makes this setting analytically tractable.

\begin{assumption}[Strict Slippage]
\label{ass:modified_trade_success}

Trade requests measure slippage relative to the post-trade \emph{spot exchange rate} of the CFMM,
not the overall exchange rate of the trade.

\end{assumption}

In other words, a trade request succeeds if and only if it would move the CFMM's reserves to some state within $L_\varepsilon(\hat p)$.
Assumption~\ref{ass:modified_trade_success} implies a pessimistic view of the trade failure probability. However, the following result signifies that Assumption~\ref{ass:modified_trade_success} is reasonable.

\begin{lemma} \label{lemma:slippage}
   At any state of the CFMMs asset reserves, the maximum size of a trade that will be successful under the model in Definition~\ref{defn:trade_success} is at most twice the maximum size of a trade that will be successful under Strict Slippage as in Assumption~\ref{ass:modified_trade_success}.
\end{lemma}

\begin{proof}
    See that the strict-quasi concavity of the trading function $f$ implies a ``convex-pricing'' of any asset. For any buy trade, the marginal exchange rate received is non-decreasing in the size of the trade. Therefore the slippage of a buy trade of $2k$ units is at least as much as the strict slippage of a  buy trade of $k$ units. A similar argument follows for sell trades.
\end{proof}

We now analyze the Markov chain over the CFMM's state, the stationary distribution of which gives us the trade failure probability under Assumption~\ref{ass:modified_trade_success}.
\begin{lemma} \label{lemma:uniform}

Let $M$ be the Markov chain defined by the state of $Y$ in the asset reserves of the CFMM with $ \mathcal Y \in L_\varepsilon(\hat p)$ and transitions 
induced by trades drawn from the distribution in Definition \ref{defn:dist_cont}. 
Under Assumption~\ref{ass:modified_trade_success}, the stationary distribution of $M$ is uniform over $L_\varepsilon(\hat p)$.

\end{lemma}

\begin{proof}
\input{proof_cont}
%
\extraqed{}
%
\end{proof}

\begin{proposition}

The probability that a trade of size $k$ units of $Y$ fails is approximately $\min(1, \frac{k}{\vert L_\varepsilon(\hat p)\vert})$, 
where the approximation error
is up to Assumption \ref{ass:modified_trade_success}.

\end{proposition}

\begin{proof}

The probability that a (without loss of generality) sell of size $k$ units of $Y$ fails is equal to the probability that a 
state $y$, drawn uniformly from the range $L_\varepsilon(p)=[y_1, y_2]$, lies in the range $[y_2-k, y_2]$.
Lemma \ref{lemma:uniform} shows this probability is $\min(1, \frac{k}{y_2-y_1})$.
\extraqed{}
\end{proof}


%% file: proof_cont.tex
It is sufficient to show that for any measurable set $A \in L_{\varepsilon}(\hat p)$, $\mu(A)=\int_{L_{\varepsilon}(\hat p)} \mu(y) p(y, A) dy$, where $\mu(\cdot)$ is a uniform measure on $L_{\varepsilon}(\hat p)$ and $p$ is a state transition kernel induced by the trade distribution. Denote a trade by $v$ wherein $v>0$ implies that the trader sells (and the CFMM buys) $|v|$ units of Y and $v <0$ implies that the trader buys (and the CFMM sells) $|v|$ units of Y.

After a trade starting from point $y \in L_{\varepsilon}(\hat p)$ the Markov chain M lands in $A$ if either $y+v\in A$ (that is, the trade succeeds)  or if $y\in A$ and $y+v\notin L_{\varepsilon}(\hat p)$  (that is, the trade fails and the initial state was in $A$). These events are mutually exclusive for any fixed trade $v$ since it can only either succeed or fail.

 Let $\mathbb P (v)$ be the probability of trade $v$ per the distribution in Definition ~\ref{defn:dist_cont}. We have the following.
\begin{align*}
&\int\limits_{L_{\varepsilon}(\hat p)}  \mu(y) p(y, A) dy \\
&= \int\limits_A \mu(y) \int\limits_{-\infty}^\infty   \mathbb P (v) (\mathds{1}(y+v \in A) + \mathds{1}(y+v \notin L_{\varepsilon}(\hat p)))   dv dy
+ \int\limits_{L_{\varepsilon}(\hat p)\setminus A}  \mu(y)  \int\limits_{-\infty}^\infty  \mathbb P (v) \mathds{1}(y+v \in A)  dv dy, \\
&= \int\limits_A \mu(y) \int\limits_{-\infty}^\infty   \mathbb P (v)  \mathds{1}(y+v \notin L_{\varepsilon}(\hat p))  dv dy
+ \int\limits_{L_{\varepsilon}(\hat p)}  \mu(y)  \int\limits_{-\infty}^\infty  \mathbb P (v) \mathds{1}(y+v \in A)  dv dy. 
\end{align*}

The second term of the above expression is the probability that the Markov chain M ends up in a state in set $A$ due to a successful trade. Since the distribution of trades $v$ is symmetric about $0$ per Definition~\ref{defn:dist_cont}, the second term of the above expression 
is equal to the probability that the Markov chain M in a state in set $A$ ends up in $L_{\varepsilon}(\hat p)$ after a successful trade. Therefore:
\begin{align*}
\int\limits_{L_{\varepsilon}(\hat p)}  \mu(y) p(y, A) dy
&= \int\limits_A \mu(y) \int\limits_{-\infty}^\infty   \mathbb P (v)  \mathds{1}(y+v \notin L_{\varepsilon}(\hat p))  dv dy
+ \int\limits_{A}  \mu(y)  \int\limits_{-\infty}^\infty  \mathbb P (v) \mathds{1}(y+v \in L_{\varepsilon}(\hat p))  dv dy, \\
&= \int\limits_A \mu(y) \int\limits_{-\infty}^\infty \mathbb P (v)  dv dy.\\
&=\mu(A). \qedhere
\end{align*}

%% file: prelim_obs_proofs.tex
\subsection{Omitted Proofs of \S \ref{sec:preliminaries} and \S \ref{sec:model}}
\label{apx:prelim_obs}

\begin{restatement*}[Observation \ref{obs:y_fn_p}]

If $f$ is strictly quasi-concave and differentiable, then for any constant $K$ and spot exchange rate $p$, 
the point $(x,y)$ where $f(x,y)=K$ and $p$ is a spot exchange rate at 
$(x,y)$ is unique.

\end{restatement*}

\begin{proof}

A constant $K=f(X_0, Y_0)$ defines a set $\lbrace x:f(x)\geq K\rbrace$.  Because $f$ is strictly quasi-concave, this set is
strictly convex.  Trades against the CFMM (starting from initial reserves $(X_0, Y_0)$) move along the boundaries of this set.
Because this set is strictly convex, no two points on the boundary can share a gradient (or subgradient).  
\extraqed{}
\end{proof}

\begin{restatement*}[Observation \ref{obs:y_fn_x}]

If $f$ is strictly increasing in both $X$ and $Y$ at every point on the positive orthant, then for a given constant function value $K$,
 the amount of $Y$ in the CFMM reserves
uniquely specifies the amount of $X$ in the reserves, and vice versa.

\end{restatement*}

\begin{proof}

If not, then $f$ would be constant on some line with either $X$ or $Y$ constant.
\extraqed{}
\end{proof}

\begin{restatement*}[Observation \ref{obs:y_nondecreasing}]
$\mathcal Y(p)$ is monotone nondecreasing.
\end{restatement*}

\begin{proof}

If $\mathcal Y(p)$ is decreasing, the level set of $f$, i.e., $\lbrace (x,y): f(x,y)\geq K\rbrace$ cannot be convex.
\extraqed{}
\end{proof}

\begin{restatement*}[Lemma \ref{lemma:y_differentiable}]
\lemmaydiffable{}
\end{restatement*}

\begin{proof}

Observation \ref{obs:y_fn_x} implies that the amount of $Y$ in the reserves can be represented as a function $\mathcal {\hat Y}(x)$ of the amount of $X$ in the reserves.
By assumption, the level sets of $f$ (other than for $f(\cdot)=0$) cannot touch the boundary of the nonnegative orthant.

Because $f$ is differentiable and increasing at every point in the positive orthant,
the map $g(x)$ from reserves $x$ to spot exchange rates at $(x, \mathcal {\hat Y}(x))$ must be a bijection from $(0, \infty)$ to $(0,\infty)$.
Because $f$ is twice-differentiable, $g(x)$ must be differentiable, and so the map $h(p)=g^{-1}(p)$ must also be differentiable.
The map $\mathcal Y(p)$ from spot exchange rates to reserves $Y$ is equal to $\mathcal{\hat{Y}}(h(p))$, and so $\hat Y(p)$ is differentiable
because $\mathcal{\hat{Y}}(\cdot)$ is differentiable and $h(\cdot)$ is differentiable.
\extraqed{}
\end{proof}

\begin{restatement*}[Lemma \ref{lemma:y_diffable}]
\lemmaypdiffable{}
\end{restatement*}

\begin{proof}
    Follows from Definitions~\ref{defn:Leps} and~\ref{defn:liquidity}. 
\end{proof}

%% file: cor_equiv_class.tex
\subsection{Omitted Proof of Corollary \ref{cor:equivclass}} \label{sec:proof_equiv}

\begin{restatement*}[Corollary \ref{cor:equivclass}]
\corequivclass{}
\end{restatement*}

\begin{proof}
Follows by substitution.  $\alpha$ rescales the derivative of the objective with respect to every variable by the same
constant and thus does not affect whether an allocation is optimal. \extraqed{}
\end{proof}

%% file: varphi_proof.tex
\subsection{Omitted Proof of Corollary \ref{cor:varphi}} \label{sec:proof_varphi}

\begin{restatement*}[Corollary \ref{cor:varphi}]
\corvarphi{}
\end{restatement*}

\begin{proof}
\begin{align*}
\int_{p_X, p_Y} \frac{\psi(p_X, p_Y)}{p_YL(p_X/p_Y)} dp_X~dp_Y
& = \int_\theta \frac{\varphi_\psi(\theta)}{L(\cot(\theta)) \sin(\theta)} d\theta\\
~& = \int_p \frac{\varphi_\psi(\theta)\sin^2(\theta)}{L(\cot(\theta)) \sin(\theta)} dp \\
~& = \int_p \frac{\varphi_\psi(\cot^{-1}(p))\sin(\cot^{-1}(p))}{L(p)}dp
\end{align*}

The first line follows by Lemma \ref{lemma:polar} (recall that $dp_X~dp_Y = r ~dr~ d\theta$), 
the second by substitution of $p=\cot(\theta)$ and $d\theta=-\sin^2(\theta)dp$ (and
changing the direction of integration --- recall $\theta=0$ when $p=\infty$),
and the third by substitution.
\extraqed{}
\end{proof}

%% file: finite_obj_proof.tex
\subsection{Omitted Proof of Lemma \ref{lemma:finite_soln}} \label{sec:prooffinitesol}
\begin{restatement*}[Lemma \ref{lemma:finite_soln}]
The optimization problem of Theorem \ref{thm:problem} always has a solution with finite objective value.
\end{restatement*}

\begin{proof}

Set $L(p)=1$ for $p\leq 1$ and $L(p)=\varphi_\psi(\cot^{-1}(p))/p^2$ otherwise.  Then 

\begin{dmath*}
\int_p \frac{\varphi_\psi(\cot^{-1}(p))\sin(\cot^{-1}(p))}{L(p)}dp \\
\leq \int_p \frac{\varphi_\psi(\cot^{-1}(p))}{L(p)}dp \\
\leq \int_0^1 \varphi_\psi(\cot^{-1}(p)) dp + \int_1^\infty \frac{dp}{p^2}
\end{dmath*}

The first term of the last line is finite, as per our assumption on trader beliefs. 

Set $Y_0=\int_0^{p_0}\frac{L(p)dp}{p}$ and $X_0=\int_{p_0}^\infty \frac{L(p)dp}{p^2}$.  Clearly both $X_0$ and $Y_0$ are finite.
Finally, rescale each $L(p)$, $X_0$, and $Y_0$ by a factor of $\frac{B}{P_X X_0 + P_Y Y_0}$ to get a new allocation $L^\prime(p)$,
$X^\prime_0$, and $Y^\prime_0$ satisfing the constraints and that still gives a finite objective value.
\extraqed{}
\end{proof}

%% file: basic_opt_proof.tex
\subsection{Omitted Proof of Lemma \ref{lemma:basicopt}} \label{proof_basicopt}

\begin{restatement*}
\lemmabasicopt{}
\end{restatement*}

\begin{proof}
The third equation
	holds	since the objective function is strictly decreasing in at least one $L(p)$ (where the belief puts a nonzero probability
	on the exchange rate $p$), so any unallocated capital
	could be allocated to increase this $L(\cdot)$ on a neighbourhood of $L(p)$ and reduce the objective.  

	The first equation holds because any unallocated units of $X$ could be allocated to $L(p^\prime)$ for a set of $p^\prime$ 
	in some neighbourhood of some $p \leq p_0$ 
	and thereby reduce the objective.  If there is no $p\leq p_0$ where the belief puts a nonzero probability, then 
	all of the capital allocated by the third constraint to $X_0$ could be reallocated into increasing $Y_0$ and thereby decreasing
	the objective.

	The second equation follows by symmetry with the argument for the first.
\extraqed{}
\end{proof}

%% file: cor_yhat_defined.tex
\subsection{Omitted Proof of Corollary \ref{cor:yhat}} \label{sec:proof_int_y}

\begin{restatement*}[Corollary \ref{cor:yhat}]

\coryhatdefined{}

\end{restatement*}

\begin{proof}

The last item of Lemma \ref{lemma:basicopt} shows that $L(p)\neq 0$ if and only if $\varphi_\psi(\cot^{-1}(p))\sin(\cot^{-1}(p))\neq 0$.
When $L(p)$ is nonzero, it is either $\sqrt{\frac{p^2}{\lambda_X}\varphi_\psi(\cot^{-1}(p))\sin(\cot^{-1}(p))}$ 
or
$\sqrt{\frac{p}{\lambda_Y}\varphi_\psi(\cot^{-1}(p))\sin(\cot^{-1}(p))}$
 (depending on the value of $p$).

By our assumption on trader beliefs, $\varphi_\psi(\cot^{-1}(p))\sin(\cot^{-1}(p))$ is a well-defined function of $p$
and is integrable.  Thus, both $\sqrt{\frac{p^2}{\lambda_X}\varphi_\psi(\cot^{-1}(p))\sin(\cot^{-1}(p))}$
and $\sqrt{\frac{p}{\lambda_Y}\varphi_\psi(\cot^{-1}(p))\sin(\cot^{-1}(p))}$ are integrable.
Monotonicity follows from $L(p)\geq 0$ and continuity from basic facts about integrals. \extraqed{}
\end{proof}

%% file: cor_lp_equiv.tex
\subsection{Omitted Proof of Corollaries \ref{cor:lp_to_equiv},~\ref{cor:inverse_problem}, and~\ref{cor:addbeliefs}} \label{sec:corr_proofs}

\begin{restatement*}[Corollary \ref{cor:lp_to_equiv}]
\corlpequiv{}
\end{restatement*}

\begin{proof}

It suffices to uniquely identify $\varphi_\psi(\cot^{-1}(p))$ for each $p$, up to some scalar.
Lemma \ref{lemma:kktobs} shows that $\varphi_\psi(\cot^{-1}(p))$ is a function of an optimal $L(p)$ and Lagrange multipliers
$\lambda_X$ or $\lambda_Y$,
and because $\lambda_X \frac{P_X}{P_Y} = \lambda_Y$, we must have that $\varphi_\psi$ is specified by $L(p)$ and $p_0$ up to some scalar
$\lambda_X$.
\extraqed{}
\end{proof}

%% file: cor_inverse_proof.tex

\begin{restatement*}[Corollary \ref{cor:inverse_problem}]
\corinverse{}
\end{restatement*}

\begin{proof}

Recall the definition of $\varphi_\psi(\cdot)$ in~\ref{cor:varphi}. 
For the given belief function $\psi,$ standard trigonometric arguments show that when
 $p\geq p_0$, we have $\varphi_\psi(\cot^{-1}(p)) = \frac{P_X L(p)^2/p}{\cos(\cot^{-1}(p))}$
 and that when $p\leq p_0$, we have $\varphi_\psi(\cot^{-1}(p)) = \frac{P_Y L(p)^2/p}{\sin(\cot^{-1}(p))}$.

Let $\hat{L}(p)$ be the allocation that results from solving the optimization problem for 
minimising the expected CFMM inefficiency for belief $\psi$.
Lemma~\ref{lemma:kktobs} part 3, gives the complementary slackness condition of $\hat{L}(p)$ and its corresponding Lagrange multiplier. 
With this, Lemma \ref{lemma:kktconds} gives the following:
 when $p\geq p_0$, $\frac{\lambda_B P_X}{p^2}=\frac{1}{\hat{L}(p)^2}(P_X L(p)^2/p)/p$,
and when $p\leq p_0$,
$\frac{\lambda_B P_Y}{p}=\frac{1}{\hat{L}(p)^2}(P_Y L(p)^2/p)$.

In other words, for all $p$,
$\lambda_B = \frac{L(p)^2}{\hat{L}(p)^2}$, so $L(\cdot)$ and $\hat{L}(\cdot)$ differ by at most a constant multiplicative factor.
But both allocations use the same budget, so it must be that $\lambda_B=1$ and $\hat{L}(\cdot) = L(\cdot)$.
\extraqed{}
\end{proof}

%% file: addbeliefs_proof.tex
\begin{restatement*}[Corollary \ref{cor:addbeliefs}]

\thmaddbeliefs{}

\end{restatement*}

\begin{proof}

Note that $\int_r \psi(r\cos(\theta), r\sin(\theta)dr$ is a linear function of each $\psi(p_X, p_Y)$,
and thus $\varphi_{\psi_1 + \psi_2}(\cdot) = \varphi_{\psi_1}(\cdot) + \varphi_{\psi_2}(\cdot)$
%
%
For any $p$ with $p\geq p_0$ and nonzero $\varphi_{\psi_1}(\cot^{-1}(p))$, 
$L_1(p)^2 = \frac{p^2}{\lambda_{1,X}} \varphi_{\psi_1}(\cot^{-1}(p))\sin(\cot^{-1}(p))$.

Similarly, for nonzero $\varphi_{\psi_2}(\cot^{-1}(p))$,
$L_2(p)^2 = \frac{p^2}{\lambda_{2,X}} \varphi_{\psi_2}(\cot^{-1}(p))\sin(\cot^{-1}(p))$.

If either $\varphi_{\psi_1}$ or $\varphi_{\psi_2}$ is nonzero at $\cot^{-1}(p)$,
then
\begin{equation*}
L(p)^2 = \frac{p^2}{\lambda_{X}} (\varphi_{\psi_1}(\cot^{-1}(p))\sin(\cot^{-1}(p)) + \varphi_{\psi_2}(\cot^{-1}(p))\sin(\cot^{-1}(p)))
\end{equation*}
Therefore,
\begin{equation*}
L(p)^2 = \frac{\lambda_{1,X}}{\lambda_X}L_1(p)^2 + \frac{\lambda_{2,X}}{\lambda_X}L_2(p)^2
\end{equation*}

An analogous argument holds for $p\leq p_0$.

The second statement follows from the fact that that when only one of $L_1(p)$ or $L_2(p)$ is nonzero,
we must have that either
$L(p)=\sqrt{\frac{\lambda_{1,X}}{\lambda_X}}L_1(p)$ or $L(p)=\sqrt{\frac{\lambda_{2,X}}{\lambda_X}}L_2(p)$.
\extraqed{}
\end{proof}

%% file: univ3_proof.tex
\subsection{Omitted Proof of Proposition \ref{prop:uniswapv3}} \label{proof_v3}

\begin{restatement*}[Proposition \ref{prop:uniswapv3}]
\univthreethm{}
\end{restatement*}

\begin{proof}

A concentrated liquidity position trades exactly as a constant product market maker within its price bounds
$p_{min}$ and $p_{max}$, and makes no trades outside of that range.

By Lemma \ref{lemma:basicopt}, the optimal $L(p)$ is $0$ for $p$ outside of the range $[p_{min}, p_{max}]$.
Inside that range, by Proposition \ref{prop:lp_dependencies}, $L(p)$ differs from the optimal liquidity allocation for
the constant product market maker by a constant, multiplicative factor (the same factor for every $p$).
Thus, the resulting liquidity allocation has the same behavior as a concentrated liquidity position. \extraqed{}
\end{proof}

%% file: balancer_proof_apx.tex
\subsection{Omitted Proof of Proposition \ref{prop:balancer}} \label{proof_balancer}

\begin{restatement*}[Proposition \ref{prop:balancer}]

\balancerthm{}

\end{restatement*}

\input{balancer_proof}

%% file: balancer_proof.tex
\begin{proof}

This trading function gives the relation $p = \frac{y\alpha}{x}$ and thus 
$\mathcal Y(p) = p^{\frac{\alpha}{\alpha+1}}(\frac{K}{\alpha^\alpha})^{\frac{1}{\alpha}}$,
for $K=f(\hat{X}, \hat{Y})$ and $\hat{X}, \hat{Y}$ is some initial state of the CFMM reserves.  

Thus, as defined by the trading function, 
$L(p) = p^{\frac{\alpha}{\alpha+1}}\frac{\alpha}{\alpha+1}\left(\frac{K}{\alpha^\alpha}\right)^{\frac{1}{\alpha+1}}$.

Corollary \ref{cor:inverse_problem} shows that a belief that leads to this liquidity allocation is
\begin{align*}
\frac{L(p)^2}{p} = p^{-1}p^{\frac{2\alpha}{\alpha+1}}\left(\frac{\alpha}{\alpha+1}\right)^2\left(\frac{K}{\alpha^\alpha}\right)^{\frac{2}{\alpha+1}} = p^{\frac{\alpha-1}.{\alpha+1}} C
\end{align*}

on the rectangle $(0, P_X] \times (0, P_Y]$ and $0$ elsewhere, for some constant $C$.  The result follows by rescaling
the belief function (Corollary \ref{cor:equivclass}).
\extraqed{}
\end{proof}

%% file: lmsr_proof.tex
\subsection{Omitted Proof of Proposition \ref{prop:lmsr}} \label{proof_lmsr}

\begin{restatement*}[Proposition \ref{prop:lmsr}]

\lmsrthm{}


\end{restatement*}

\begin{proof}

This trading function implies the relationship $p=e^{y-x}.$ 

Combining this with the equation
$e^{-y} + e^{-x}=K$ (for some constant $K$) gives 
$(1+p)e^{-y}=K$
and thus $\mathcal Y(p)=\ln(\frac{1+p}{K})$.
From the definition of liquidity, we obtain $L(p)=\frac{p}{1+p}$.

Corollary \ref{cor:inverse_problem} shows that a belief function that leads to this 
liquidity allocation is (with $p=p_X/p_Y$)
\begin{equation*}
\frac{L(p)^2}{p}= \frac{p}{(1+p)^2} = \frac{p_X p_Y}{(p_X+p_Y)^2}
\end{equation*}
for $(p_X, p_Y)\in (0, P_X] \times (0, P_Y]$ and $0$ elsewhere.
%
The result follows by rescaling the belief function (Corollary \ref{cor:equivclass}).
\extraqed{}
\end{proof}

%% file: loss_proof.tex
\subsection{Omitted Proof of Proposition \ref{prop:divergence}} \label{sec:expected_future_value}

\begin{restatement*}[Proposition \ref{prop:divergence}]

\thmdivergence{}

\end{restatement*}

\begin{proof}

\begin{dmath*}
 \frac{1}{N_{\psi}} \int\limits_{p_X, p_Y}\psi(p_X, p_Y)\left(p_X \mathcal X(p_X/p_Y)) + p_Y \mathcal Y(p_X/p_Y)\right)dp_X~dp_Y \\
= \frac{1}{N_{\psi}} \int\limits_{p_X, p_Y}\psi(p_X, p_Y)\left(p_X \int\limits_{p_X/p_Y}^\infty\frac{L(p)}{p^2} dp + p_Y\int\limits_0^{p_X/p_Y}\frac{L(p)}{p} dp\right)dp_X~dp_Y \\
= \frac{1}{N_{\psi}} \int\limits_{0}^\infty \left( \frac{L(p)}{p^2} \int\limits_{p_X, p_Y}p_X\psi(p_X, p_Y)\mathds{1}_{\lbrace \frac{p_X}{p_Y} \leq p\rbrace}dp_X~dp_Y
+ \frac{L(p)}{p} \int\limits_{p_X, p_Y}p_Y\psi(p_X, p_Y)\mathds{1}_{\lbrace \frac{p_X}{p_Y} \geq p\rbrace}dp_X~dp_Y\right)dp
\end{dmath*}

The first equation follows by substitution of the equations in Observation \ref{obs:alloc}.

Note that for any $p_X, p_Y$, the term
$\frac{L(p)}{p^2}$ for any $p$ appears in the integral $\int_{p_X/p_Y}^\infty\frac{L(p)dp}{p^2}$ if and only if 
$p \geq p_X/p_Y$.  The result follows from rearranging the integral to group terms by $L(p)$.
\extraqed{}
\end{proof}

%% file: bid_ask_proof.tex
\subsection{Omitted Proof of Theorem \ref{thm:qual_divloss}} \label{sec:proof_div}
\begin{restatement*}[Theorem \ref{thm:qual_divloss}]

\thmdivloss{}

\end{restatement*}

\begin{proof}

Define $\varphi^\prime(\theta) = \delta\int_r rate_{\delta}(r\cos(\theta), r\sin(\theta)\psi(r\cos(\theta), r\sin(\theta))dr$.

The KKT conditions for the first problem give the following (nearly identically to those in Lemma \ref{lemma:kktconds}, just using $L_1(\cdot)$ in place of $L(\cdot)$):
\begin{enumerate}
	\item
		For all $p$ with $p\geq p_0$, 
		$\frac{\lambda_X}{p^2} = \frac{1}{L_1(p)^2}\varphi^\prime(\cot^{-1}(p))\sin(\cot^{-1}(p)) + \lambda_{L_1(p)}$.

	\item 
		For all $p$ with $p\leq p_0$, 
		$\frac{\lambda_Y}{p} = \frac{1}{L_1(p)^2}\varphi^\prime(\cot^{-1}(p))\sin(\cot^{-1}(p)) + \lambda_{L_1(p)}$.

	\item $ \lambda_X = P_X \lambda_B$ and  $ \lambda_Y = P_Y \lambda_B$.
\end{enumerate}

Observe that the derivative, with respect to $L(p)$, of the divergence loss, is
\begin{equation*}
\kappa(p) 
= \frac{1}{N_\psi} \left( \frac{1}{p^2} \iint\limits_{p_X, p_Y}p_X\psi(p_X, p_Y)\mathds{1}_{\lbrace \frac{p_X}{p_Y} \leq p\rbrace}dp_X~dp_Y
+ \frac{1}{p} \iint\limits_{p_X, p_Y}p_Y\psi(p_X, p_Y)\mathds{1}_{\lbrace \frac{p_X}{p_Y} \geq p\rbrace}dp_X~dp_Y\right).
\end{equation*}

Computing the KKT conditions for the second problem gives the following:
\begin{enumerate}
	\item
		For all $p$ with $p\geq p_0$, 
		$\frac{\lambda_X}{p^2} = \frac{1}{L_2(p)^2}\varphi^\prime(\cot^{-1}(p))\sin(\cot^{-1}(p)) + \kappa(p) +\lambda_{L_2(p)}$.

	\item 
		For all $p$ with $p\leq p_0$, 
		$\frac{\lambda_Y}{p} = \frac{1}{L_2(p)^2}\varphi^\prime(\cot^{-1}(p))\sin(\cot^{-1}(p)) + \kappa(p) + \lambda_{L_2(p)}$.
\end{enumerate}


As defined in the theorem statement, $X_1=\int_{p_0}^\infty \frac{L_1(p)}{p^2}dp$ and $X_2=\int_{p_0}^\infty \frac{L_2(p)}{p^2}dp$ 
are the optimal initial quantities of $X$.  
Normalizing $L_1$ and $L_2$ by $X_1$ and $X_2$ respectively gives the equation
\begin{equation}
\label{eqn:normalized}
\int_{p_0}^\infty \frac{L_1(p)}{X_1p^2} dp=\int_{p_0}^\infty \frac{L_2(p)}{X_2p^2} dp
\end{equation}  

This implies that, when normalized by $X_1$ and $X_2$, $p\geq p_0$, and $\varphi^\prime(\cot^{-1}(p))\neq 0$, we have that
\begin{dmath*}
L_2(p)
= \sqrt{\frac{\varphi^\prime(\cot^{-1}(p))\sin(\cot^{-1}(p))}{\frac{\lambda_X}{p^2} - \kappa(p)}} \\
= \sqrt{\frac{\varphi^\prime(\cot^{-1}(p))\sin(\cot^{-1}(p))}{\frac{\lambda_X}{p^2}}*\frac{\frac{\lambda_X}{p^2}}{\frac{\lambda_X}{p_2} - \kappa(p)}} \\
= L_1(p)\frac{X_2}{X_1} \sqrt{\frac{\lambda_X}{\lambda_X - p^2\kappa(p)}}
\end{dmath*}

A similar argument shows that when $p \leq p_0$,
\begin{dmath*}
L_2(p)
= L_1(p)\frac{Y_2}{Y_1} \sqrt{\frac{\lambda_Y}{\lambda_Y - p\kappa(p)}}
\end{dmath*}

\newcommand{\psitp}
{
\psi(r\cos(\theta^\prime), r\sin(\theta^\prime))
}
\newcommand{\drdt}
{
dr ~d\theta^\prime
}

Arithmetic calculation gives that 
\begin{dmath*}
\kappa(p)p^2 = \frac{1}{N_{\psi}}\int_r \int_{\theta^\prime=\theta}^{\pi/2} r^2 \cos(\theta^\prime)\psitp{}\drdt{}
+ \cot(\theta) \int_r \int_{\theta^\prime=0}^\theta r^2\sin(\theta^\prime)\psitp{}\drdt{}
\end{dmath*}
and thus that
\begin{equation*}
\frac{d(\kappa(\theta)\cot(\theta)^2)}{d\theta} = - \frac{\csc^2(\theta)}{N_{\psi}}\int_r\int_{\theta^\prime=0}^\theta r^2\sin(\theta^\prime)\psitp{}\drdt{} \leq 0
\end{equation*}

$\kappa(\theta)\cot(\theta)^2$ is therefore decreasing in $\theta$, so $\kappa(p)p^2$ is increasing in $p$ (since $p = \cot(\theta)$).
Therefore, $\sqrt{\frac{\lambda_X}{\lambda_X - p^2\kappa(p)}}$ increases as $p$ goes to $\infty$.

By an analogous argument,
\begin{dmath*}
\kappa(p)p = \frac{\tan(\theta)}{N_{\psi}} \int_r \int_{\theta^\prime=\theta}^{\pi/2} r^2 \cos(\theta^\prime)\psitp{}\drdt{}
+ \int_r \int_{\theta^\prime=0}^\theta r^2\sin(\theta^\prime)\psitp{}\drdt{}
\end{dmath*}
and thus
\begin{equation*}
\frac{d(\kappa(\theta)\cot(\theta))}{d\theta} = \frac{\sec^2(\theta)}{N_{\psi}}  \int_r\int_{\theta^\prime=\theta}^{\pi/2} r^2\sin(\theta^\prime)\psitp{}\drdt{} \geq 0
\end{equation*}

$\kappa(\theta)\cot(\theta)$ is therefore increasing in $\theta$, so $\kappa(p)p$ increases as $p$ decreases.

Therefore, $\sqrt{\frac{\lambda_Y}{\lambda_Y - p\kappa(p)}}$ increases as $p$ goes to $0$.


Equation 
\ref{eqn:normalized} implies that the quantities $\frac{L_1(p)}{X_1p^2}$ and $\frac{L_2(p)}{X_2p^2}$ integrate to the same value, but $L_2(\cdot)$ increases
strictly more quickly than $L_1(\cdot)$, so there must be a point 
$p_1>p_0$ beyond which $\frac{L_1(p)}{X_1} \leq \frac{L_2(p)}{X_2}$.

An analogous argument holds for $p<p_0$.
\extraqed{}
\end{proof}